\documentclass[11pt]{article}
\usepackage{pgfplots}
\pgfplotsset{compat=1.7}
\usepackage{subcaption}
\usepackage[margin=1in,letterpaper]{geometry}
\usepackage[T1]{fontenc}
\usepackage{lmodern}
\usepackage{amsmath}
\usepackage{amssymb}
\usepackage{amsthm}
\usepackage{color,soul}
\usepackage{xcolor}
\usepackage{graphicx}
\usepackage{algorithm}
\usepackage{tcolorbox}
\usepackage[normalem]{ulem}

\usepackage{tikz}
\usetikzlibrary{patterns}
\usepackage{amsopn,amssymb,amsthm,amsmath,amsthm}
\usepackage[noend]{algpseudocode}

\theoremstyle{plain}
\newtheorem{thm}{\protect\theoremname}
\theoremstyle{plain}
\newtheorem{claim}[thm]{\protect\claimname}
\theoremstyle{plain}

\theoremstyle{plain}
\newtheorem{lem}[thm]{\protect\lemmaname}
\theoremstyle{plain}

\theoremstyle{definition}
\newtheorem{defn}[thm]{\protect\definitionname}
\theoremstyle{definition}

\theoremstyle{definition}
\newtheorem{rem}{\protect\remarkname}
\theoremstyle{plain}

\makeatother

\providecommand{\claimname}{Claim}
\providecommand{\lemmaname}{Lemma}
\providecommand{\propositionname}{Proposition}
\providecommand{\theoremname}{Theorem}
\providecommand{\corollaryname}{Corollary}
\providecommand{\definitionname}{Definition}
\providecommand{\assumptionname}{Assumption}
\providecommand{\remarkname}{Remark}

\global\long\def\RR{\mathbb{R}}

\global\long\def\R{{\cal R}}

\global\long\def\N{{\cal N}}
\global\long\def\e{{\mathbf{e}}}

\global\long\def\q{{\mathbf{q}}}

\global\long\def\p{{\mathbf{p}}}

\global\long\def\poly#1{\mathrm{poly}\left(#1\right)}

\global\def\eqdef{\equiv}

\newcommand{\bv}[1]{\mathbf{#1}}

\newcommand{\Rbb}{\mathbb{R}}

\newcommand{\ignore}[1]{}
\usepackage{comment}

\title{A Quasi-Monte Carlo Data Structure for Smooth Kernel Evaluations}
\author{Moses Charikar\thanks{Supported by a Simons Investigator award.}\\Stanford University \and Michael Kapralov\thanks{Supported in part by ERC Starting Grant 759471.}\\EPFL \and Erik Waingarten\thanks{Part of this work was done as a postdoc at Stanford University, supported by an NSF postdoctoral fellowship and Moses Charikar's Simons Investigator award.}\\University of Pennsylvania}

{\makeatletter
	\gdef\xxxmark{%
		\expandafter\ifx\csname @mpargs\endcsname\relax 
		\expandafter\ifx\csname @captype\endcsname\relax 
		\marginpar{xxx}
		\else
		xxx 
		\fi
		\else
		xxx 
		\fi}
	\gdef\xxx{\@ifnextchar[\xxx@lab\xxx@nolab}
	\long\gdef\xxx@lab[#1]#2{{\bf [\xxxmark #2 ---{\sc #1}]}}
	\long\gdef\xxx@nolab#1{{\bf [\xxxmark #1]}}
}

\newcommand{\E}{\mathbb{E}}

\renewcommand{\e}{\epsilon}

\renewcommand{\p}{\mathbf{p}}
\renewcommand{\q}{\mathbf{q}}

\begin{document}

\maketitle

\begin{abstract}
In the kernel density estimation (KDE) problem one is given a kernel $K(x, y)$ and a dataset $P$ of points in a high dimensional Euclidean space, and must prepare a small space data structure that can quickly answer density queries: given a point $q$, output a $(1+\e)$-approximation to $\mu:=\frac1{|P|}\sum_{p\in P} K(p, q)$. The classical approach to KDE (and the more general problem of matrix vector multiplication for kernel matrices)  is the celebrated fast multipole method of Greengard and Rokhlin [1983]. The fast multipole method combines a basic space partitioning approach with a multidimensional Taylor expansion, which yields a $\approx \log^d (n/\e)$ query time (exponential in the dimension $d$).  A recent line of work initiated by Charikar and Siminelakis [2017] achieved polynomial dependence on $d$ via a combination of random sampling and randomized space partitioning, with Backurs et al. [2018]  giving an efficient data structure with query time $\approx \poly{\log(1/\mu)}/\e^2$ for smooth kernels. 

Quadratic dependence on $\e$, inherent to the sampling (i.e., Monte Carlo) methods above, is prohibitively expensive for small $\e$. This is a classical issue addressed by quasi-Monte Carlo methods in numerical analysis. The high level idea in quasi-Monte Carlo methods is to replace random sampling with a discrepancy based approach -- an idea recently applied to coresets for KDE by Phillips and Tai [2020]. The work of Phillips and Tai gives a space efficient data structure with query complexity $\approx 1/(\e \mu)$. This is polynomially better in $1/\e$, but exponentially worse in $1/\mu$. In this work we show how to get the best of both worlds: we give a data structure with $\approx \poly{\log(1/\mu)}/\e$ query time for smooth kernel KDE. Our main insight is a new way to combine discrepancy theory with randomized space partitioning inspired by, but significantly more efficient than, that of the fast multipole methods. We hope that our techniques will find further applications to linear algebra for kernel matrices. 
\end{abstract}

\setcounter{page}{0}

\newpage 

\tableofcontents

\setcounter{page}{0}

\newpage

\newcommand{\sfK}{\mathsf{K}}
\newcommand{\Nbb}{\mathbb{N}}
\newcommand{\calQ}{\mathcal{Q}}
\newcommand{\rout}{r_{\mathrm{out}}}
\newcommand{\rin}{r_{\mathrm{in}}}
\newcommand{\rmin}{r_{\min}}

\section{Introduction}

In the \emph{kernel evaluation} problem, an algorithm receives as input a dataset $P$ of $n$ points $p_1, \dots, p_n \in \RR^d$ and must preprocess it into a small-space data structure that allows one to quickly approximate, given a query $q\in \RR^d$, the quantity
\begin{align}\label{eq:kde-intro}
\sfK(P,q):=\sum_{p\in P}\sfK(p,q).
\end{align}
where $\sfK(p,q)$ is the kernel function. In this paper, we will study positive definite (p.d) radial kernels $\sfK$. In particular, we consider kernel functions $\sfK \colon \Rbb^d \times \Rbb^d \to \Rbb_{\geq 0}$ given by a decreasing function of the distance, i.e., there exists a function $G \colon \Rbb_{\geq 0} \to \Rbb_{\geq 0}$ such that $\sfK(p, q) = G(\|p-q\|_2^2)$. The restriction that $\sfK$ is positive definite means that for any collection of points $x_1,\dots, x_m \in \Rbb^d$, the $m\times m$ kernel matrix whose $(i,j)$-entry is $\sfK(x_i, x_j)$ is positive definite. Some example of prominent p.d radial kernels include\footnote{The Gaussian kernel $\sfK(p,q) = \exp(-\|p-q\|_2^2/(2\sigma^2)$) is a widely-used p.d kernel, but not ``smooth'' (as per the definition in~\cite{BCIS18}).} 
the Cauchy kernel (or, more generally, the rational quadratic kernel) for any $\beta \geq 1$, and the Student-$t$ kernel, respectively,
\begin{align}\label{eq:smooth-ks}
\sfK(p,q):=\left(\frac{1}{1+\|p-q\|_2^2}\right)^{\beta} \qquad \text{and}\qquad \sfK(p, q) := \dfrac{1}{1 + \|p-q\|_2^t}.
\end{align}
A variety of kernels are used in applications~\cite{shawe2004kernel,RasmussenW06}, and kernel methods are a fundamental approach with numerous applications in machine learning, statistics and data analysis~\cite{fan1996local,scholkopf2001learning,joshi2011comparing,schubert2014generalized,genovese2014nonparametric,arias2015estimation,gan2017scalable}.

For example, {kernel density estimation} is a classic tool in non-parametric statistics, where a kernel function is applied to extrapolate a function specified on a discrete set of points to the entire space.  This is used in algorithms for {mode estimation}~\cite{GSM03},
{outlier detection}~\cite{schubert2014generalized}, {density based clustering}~\cite{rinaldo2010generalized} and other problems. 
Kernel methods (e.g., regression~\cite{szabo2016learning}) have also been applied to objects specified by point-clouds or distributions~\cite{gartner2002multi}, 
requiring summing up the kernel function between all pairs of points across two sets.
Another application is kernel mean estimation  using an empirical average (see ~\cite{muandet2017kernel}, section 3.4.2).

We are interested in fast and space efficient approximations algorithms for kernel evaluations. An estimator for $\hat{\sfK}(P, q)$ is a $(1+\e)$-approximation to $\sfK(P, q)$ if 
\begin{align*}
(1-\e) \cdot \sfK(P, q)\leq \hat{\sfK}(P, q)\leq (1+\e) \cdot \sfK(P, q).
\end{align*}

The kernel density estimation problem has received a lot of attention over the years, with a number of powerful algorithmic ideas leading to different precision/space/query time tradeoffs. The focus of this work is algorithms for the so-called ``high-dimensional'' regime: we will study algorithms whose complexity will depend at most polynomially, as opposed to exponentially in the underlying dimension. 

\paragraph{Prior Work: Fast Multipole Methods.} The celebrated fast multipole method \cite{fast-multipole}  can be used to obtain efficient data structure for kernel evaluation (see also the Barnes-Hut algorithm~\cite{barnes-hut}). However, this approach suffers from an exponential dependence on the dimensionality of the input data points: it provides a $(1+\e)$-approximation to kernel evaluation value using space $\approx n \log^{d} (n/\e)$  and query time $\approx \log^d (n/\e)$. 

In particular, the exponential dependence on the dimensionality is due to a (deterministic) space partitioning procedure (essentially building a quadtree) which is central to the fast-multipole method. More generally, this deficiency is shared by other tree-based methods for kernel density estimation~\cite{gray2001n,gray2003nonparametric,yang2003improved,lee2006dual,ram2009linear}. Methods based on polynomial approximations have recently been used to obtain fast algorithms for kernel graphs~\cite{AlmanCS020} as well as attention mechanisms in transformers~\cite{DBLP:journals/corr/abs-2302-13214} -- all these approaches are only efficient in relatively low dimensions. 

\paragraph{Prior Work: Sampling-based approaches (Monte-Carlo methods).} A recent line of work~\cite{charikar2017hashing,charikar2019multi-resolution,backurs2018efficient,backurs2019space,DBLP:conf/focs/CharikarKNS20} sought to design sublinear query-time algorithms for kernel density estimation while avoiding exponential dependencies in the dimension (thereby allowing these methods to be scaled to high-dimensional spaces).  These works parametrize the query time of the data structure in terms of the value $\mu = \sfK(P, q)/n$, and the goal is to achieve query times which are significantly faster than $O(d/(\e^2 \mu))$, which is what one achieves via uniform random sampling. Surprisingly,~\cite{CS17} showed how, for several kernels, one can reduce the query time to $O(d / (\e^2\sqrt{\mu}))$ by using the Locality-Sensitive Hashing framework of Indyk and Motwani~\cite{IndykM98}. Furthermore, \cite{backurs2018efficient} (and also ~\cite{DBLP:conf/focs/CharikarKNS20}) developed the approach further, and showed that for ``smooth'' kernels, i.e. kernels with polynomial decay, a multiplicative approximation can be obtained with just a {\em polylogarithmic dependence on $1/\mu$}. Specifically, they achieve a $(1+\e)$-approximation using space $\approx n \cdot \poly{\log (1/\mu)}/\e^2$ space and 
\begin{equation}\label{eq:query-sampling}
\text{query time}\approx \poly{\log (1/\mu)}/\e^2.
\end{equation}

\paragraph{Prior Work: Quasi-Monte Carlo (i.e., discrepancy-based) methods.} In the approaches above, the focus has always been on the dependence on $\mu$, and the additional factor of $1/\e^2$ a consequence of the sampling-based approach. More broadly, this quadratic dependence on $1/\e$ is unsurprising. It is common, in approaches via random sampling, to design an estimator and utilize Chebyshev's inequality to upper bound the probability that the estimator deviates from its expectation. On the other hand, in applications where the $1/\e^2$ dependence (coming from random sampling) is prohibitively expensive, a common approach in numerical analysis is to use quasi-Monte Carlo methods.

At a high level, these techniques are based on discrepancy theory. They seek to find a ``random-like'' set of samples (such that the estimator will approximate the expectation), and at the same time, minimize the ``random deviations'' that one expects from random sampling.
 In the context of kernel density estimation, this idea was used in the beautiful work of Phillips and Tai~\cite{PT20}, who used discrepancy theory to build a small coreset for kernel density estimation. As we detail in Section~\ref{sec:tech-overview}, standard subsampling approaches correspond to chosing the colors uniformly at random, but \cite{PT20} show that a coloring with much lower discrepancy can be constructed using Banaszcyk's theorem. Repeatedly `subsampling' using these low-discrepancy colorings, they give a data structure with space $\approx \poly{d}/(\e \mu)$ and
\begin{equation}\label{eq:query-coresets}
\text{query time} \approx \poly{d}/(\e \mu).
\end{equation}
This gives yet another way to improve on the $O(d/(\e^2\mu))$ query time that one achieves via uniform random sampling.\footnote{One should interpret the dimensionality $d$ as being $\omega(\log n) \leq d \leq n^{o(1)}$. This means that exponential dependencies on $d$ should be avoided (as they incur super-polynomial factors in $n$), but arbitrary polynomial dependencies are allowed.} In the work of~\cite{PT20}, the focus is on the dependence on $\e$ and the dependence on $1/\mu$ remains linear.\footnote{In fact, the data structure of~\cite{PT20} has a Las Vegas guarantee; it uses randomness to build the data structure, but then guarantees that all queries are correct.} 

\paragraph{Our contribution.} In this work, we show how to achieve the best of both~\eqref{eq:query-sampling} and~\eqref{eq:query-coresets} quantitatively. Namely, we show how one may use randomized partitioning techniques (like those in~\cite{BCIS18, DBLP:conf/focs/CharikarKNS20}) to obtain a poly-logarithmic dependence on $1/\mu$ for smooth kernels, and at the same time, a linear dependence on $1/\e$. The resulting data structure will obtain space complexity $\approx n \cdot \poly{d \log(1/\mu)} / \e$ and 
\[ \text{query time } \approx \poly{d\log(1/\mu)} / \e.  \]
At a qualitative level, our work gives new structural results and algorithmic techniques for both the sampling-based and quasi-Monte Carlo-based approaches, and we are hopeful that these techniques will prove useful for fast algorithms in related tasks for kernels in high-dimensional spaces. 

From the perspective of the sampling-based methods, our work shows that the quadratic dependence on $1/\e$ is not intrinsic to the randomized approaches for high-dimensions, and that quasi-Monte Carlo (discrepancy-based) techniques can be used to design kernel density estimators in high-dimensions. From the perspective of the quasi-Monte Carlo methods, our work shows that, if one allows randomized data structures, then randomized space partitioning can give exponential improvements on the $\mu$-dependence of discrepancy for smooth kernels; from $\poly{d}/(\e \mu)$ to roughly $\poly{ d\log(1/\mu)}/\e$. In what follows, we will formally state our results and some open problems, and give an outline of our techniques.


\paragraph{Our results.} Our focus is on developing fast data structures for ``smooth'' p.d radial kernels. We reproduce the definition of ``smooth'' kernel from \cite{BCIS18} below, and then we state our main result. As we will soon see, the smoothness condition will become important in the technical details. After discussing the main result, we will give a few open problems and highlight a few concrete challenges involved in obtaining improved $\e$-dependencies for non-smooth kernels (like the Gaussian kernel).

\begin{defn}[Smooth Function \cite{BCIS18}]\label{dfn:smooth}
A kernel $\sfK \colon \Rbb^d \times \Rbb^d \to \Rbb$ is $(L, t)$-smooth if for any three points $p_1, p_2, q \in \RR^d$ with $p_1 \neq q \neq p_2$, 
\begin{align*}
\max \left\{ \dfrac{\sfK(p_1, q)}{\sfK(p_2, q)},  \dfrac{\sfK(p_2, q)}{\sfK(p_1, q)} \right\} \leq L \cdot \left( \max\left\{\dfrac{\|p_1 - q\|_2}{\|p_2-q\|_2}, \dfrac{\|p_2-q\|_2}{\|p_1 - q\|_2} \right\}\right)^{t}.
\end{align*}
\end{defn}

\begin{rem}
We remark that the definition above encompasses kernels with a polynomial decay, such as, for example, the rational quadratic kernel and its variants with polynomial decay~\eqref{eq:smooth-ks}. While certainly less popular than the Gaussian kernel, the rational quadratic kernel is commonly listed among the standard covariance functions (i.e. kernels) in the literature on Gaussian processes. See, e.g., Section 4.2.1 of~\cite{RasmussenW06}, as well as Section 5.4 of the same book, which in particular presents settings where the rational quadratic covariance assumption leads to improvements of the Gaussian. 
\end{rem}

\ignore{Examples of smooth kernels include polynomial kernels of the form
$$
K(\p, \q)=\left(\frac1{1+\|\p-\q\|^2}\right)^t
$$
for some integer $t\geq 1$.

Our main result is}
\begin{thm}[Main result (Informal -- see Theorem~\ref{thm:main-thm})]\label{thm:main}
For every $\e\in (0, 1)$ and $\mu>0$ there exists a randomized data structure for kernel evaluation of $(L, t)$-smooth p.d radial kernels $\sfK$ with polynomial preprocessing time, 
\[ \text{space complexity: } n \cdot L \cdot (d\log(n\Phi/(\e\mu)))^{O(t)} /\e \qquad \text{and}\qquad\text{query time: } L \cdot (d \log(n\Phi/(\e\mu)))^{O(t)} / \e \]
which outputs a multiplicative $(1\pm \e)$-approximation to kernel evaluations whenever the kernel evaluation is at least $\mu n$, where $\Phi$ is the aspect ratio of the points.\footnote{The dependence on $\Phi$ seems necessary when making no assumptions on the smooth kernels. This dependence can be removed under a mild assumption on the decay of the kernel (which holds for the Cauchy and Student-$t$ kernels listed as examples).}
\end{thm}

The main conceptual contribution behind the result is a framework for combining discrepancy techniques with randomized space partitioning. To the best of our knowledge, ours is the first work that succeeds in combining these lines of work. At a technical level the two main innovations that we introduce are {\bf (a)} a strong bound on the discrepancy of the kernel matrix for smooth kernels under a geometric separability assumption akin to the one used in Fast Multipole Methods and {\bf (b)} a way to partition a dataset (and potential queries) into a small number of pieces that ensure separation. 

\paragraph{Dependence on $\Phi$.} We note that a dependence on the aspect ratio $\Phi$ of the dataset is also present in~\cite{BCIS18}, and seems necessary for general smooth kernels. For instance, a natural approach to removing it would be to {\bf (a)} remove data points that are $\poly{1/(\e \mu)}$ far from the query point from consideration and {\bf (b)} discretize points without significantly changing the kernel values but ensure that non-equal points have a minimum distance---{\bf (a)} and {\bf (b)} effectively upper bounds the aspect ratio $\Phi$ by upper bounding the maximum distance and lower bounding the minimum distance. However, this does not work without a closer look at the specific kernel function $\sfK$. For example, the counting kernel $\sfK(x, y) = 1$ for all $x,y$ would mean one cannot do {\bf (a)}. For {\bf (b)}, consider the kernel $\sfK(x, y)=\frac1{1+\sigma^2 \|x-y\|^2}$, where we let $\sigma$ go to infinity. This kernel is smooth as per our definition for every $\sigma$, independent of $L$. At the same time, the value of the kernel starts dropping sharply at $\|x-y\|_2\approx 1/\sigma$, so one should not give a discretization independent of $\sigma$. At the same time, the dependence on $\Phi$ can be removed for specific kernels such as the Cauchy or the $t$-Student kernel using the approach outlined above.

\paragraph{Dependence on $(d \log (n\Phi/(\e\mu)))^{O(t)}$.} Our techniques will naturally incur a poly-logarithmic dependence on $d \log (n\Phi/(\e\mu)$, where the specific power of the exponent will be $O(t)$ for $(L, t)$-smooth kernels $\sfK$. For the Cauchy kernel or Student-$t$ kernel with $\beta$ and $t$ being $O(1)$, the additional factors are poly-logarithmic. Interestingly, the dependence on the smoothness $t$ in~\cite{BCIS18} is $2^{O(t)}$, which becomes important once $t = \Theta(\log n)$. 

\subsection{Future Directions and Open Problems} We hope that our techniques, which allow one to combine discrepancy methods with randomized space partitioning,  will prove instrumental in resolving other exciting questions in numerical linear algebra with kernel matrices. We mention two prominent questions here.

\begin{itemize}
\item \underline{Kernel density estimation for Gaussian Kernel.} Does there exist a data structure for kernel density estimation for the Gaussian kernel (i.e., $\sfK(p, q) = \exp(-\|p-q\|_2^2/2)$) with polynomial space complexity and query time $\poly{d} / (\e \cdot \mu^{0.99})$? The Gaussian kernel is not ``smooth'' so the approach from this paper would degrade to $(d \log(1/\mu))^{\Omega(d)} / \e$. A specific hurdle is that our partitioning techniques only ensure an $\Omega(1/\sqrt{d})$ relative separation of the query and dataset (here `relative' refers to the size of a bounding Euclidean ball - see Section~\ref{sec:tech-overview} for more details). However, a natural extension of our data-dependent embeddings to the Gaussian kernel incurs an exponential dependence on the inverse of this relative separation -- see Remark~\ref{rem:smoothness} for more details.

\item \underline{Fast Multipole Methods with a polynomial dependence on dimension.} In fast multipole methods one approximates a kernel matrix by first using a crude ($\ell_\infty$ ball carving) space partitioning to partition space into bounded regions, and then Taylor expands the kernel arounds centers of these regions to approximate  interactions between well-separated data points (see, e.g.,~\cite{fast-multipole}).  Both steps incur an exponential in the dimension loss: the former because a constant factor separation is ensured in the crude $\ell_\infty$ ball-carving used, which requires breaking an $\ell_\infty$ ball into $d^{\Omega(d)}$ balls of factor $\sqrt{d}$ smaller radius. The latter because a Taylor expansion has an exponential number of terms in the dimension of the dataset. The latter exponential loss was recently overcome by the work of~\cite{ahle2019oblivious}, who showed how to {\em sketch} the polynomial kernel (i.e., the Taylor expansion) with only a polynomial in $d$ dependence. However, the approach of~\cite{ahle2019oblivious} suffers from a polynomial dependence on the radius of the dataset, as they do not supply a corresponding space partitioning primitive. Our space partitioning methods are able to exploit only a $\Omega(\frac1{\sqrt{d}})$ relative separation, and only result in a polynomial in $d$ dependence: can our techniques be used together with~\cite{ahle2019oblivious} to optimally sketch kernel matrices?

\if  Consider a smooth, p.d radial kernel $\sfK \colon \RR^d \times \RR^d \to [0, 1]$. Does there exist an algorithm which receives a dataset of $n$ points $a_1,\dots, a_n \in \RR^d$, a label vector $b \in \RR^n$, and a regularization parameter $\lambda > 0$ and can output in time $n \cdot \poly{d, \log(n), 1/\e, 1/\lambda}$ a vector $\hat{x} \in \RR^n$ which is a $(1+\e)$-approximate minimizer of $\| \sfK x - b \|_2^2 + \lambda \cdot \langle \sfK x, x \rangle$?
\fi
\end{itemize}

\ignore{
Should we try to formulate the same approach we did for the Gaussian kernel as an open problem? I think this may alleviate a reviewer who just says that we don't improve on all regimes, and that our thing doesn't apply to Gaussian kernels... At the very least, we could highlight the fact that there are signficant challenges in going beyond smooth kernels here. \xxx[MK]{yes, sounds good!}\xxx[MC]{agreed}}

\ignore{\begin{rem}
We note that, in contrast to approaches to KDE based on random sampling, our data structure can be made Las Vegas, i.e. it can be made to report a failure whenever it fails as opposed to just return an incorrect answer. \xxx[MK]{details?}
\end{rem}}


\section{Technical Overview}\label{sec:tech-overview}
\newcommand{\Sbb}{\mathbb{S}}
\newcommand{\Zbb}{\mathbb{Z}}

In this section, we give an overview of the techniques involved. In order to highlight our main contributions, it will be useful to consider the case of the 2-Student kernel $\sfK\colon \Rbb^d \times \Rbb^d \to [0, 1]$ for concreteness. This kernel is given by
\[ \sfK(x, y) = \dfrac{1}{1 + \|x-y\|_2^2}. \]
 One receives a dataset $P \subset \Rbb^{d}$ of points and seeks to process them into a data structure to support kernel evaluation queries. Namely, a query $q \in \Rbb^{d}$ will come, and we want to estimate $\sum_{p \in P} \sfK(p, q)$ up to a factor of $1\pm \e$, and similarly to \cite{charikar2017hashing,charikar2019multi-resolution,backurs2018efficient,backurs2019space,DBLP:conf/focs/CharikarKNS20} we will parametrize our time and space complexity in terms of $\mu = (1/n)\sum_{p \in P}\sfK(p, q)$. 
 
We start by explaining two prior approaches which will be important ingredients in our scheme. First, a simple random sampling approach, and then the prior work of~\cite{PT20} which shows how to use discrepancy theory to obtain an improved dependence on $\e$. Then, we overview our approach. First, we show how we may improve on the discrepancy bounds when datasets are ``well-separated'' from a query, and then how to algorithmically utilize the improved discrepancy bounds in the worst-case. 
 
\paragraph{Random sampling as repeated coloring.} It is not hard to see that, since kernel values are always between 0 and 1, a uniformly random sample from $P$ of size $O(1/(\e^2 \mu))$ will approximate the kernel evaluation of any $q$ with probability at least $0.9$. More generally, one may take samples from an unbiased estimator of $(1/n)\sum_{p \in P} \sfK(p, q)$, and show that the estimate is good via bounding the variance of the estimator. If proceeding with this plan, the quadratic dependence on $1/\e$ is a consequence of using Chebyshev's inequality; since the probability that the estimate is off by more than $\e \mu$ (the quantity we want to minimize by taking more samples) becomes at most the variance divided by $\e^2 \mu^2$. To facilitate comparison with discrepancy-based approaches, we now sketch an alternative derivation of this result. And for this purpose it will be useful to instead think of repeatedly sampling data points in $P$ with probability $1/2$ and analyzing how errors in the corresponding KDE estimates accumulate.

Suppose that we would like to subsample the dataset $P$ to a subset $P'$ containing about half of the points in $P$ while preserving KDE value. It is convenient to think of this process as coloring the points in $P$ 
 $$
 \chi \colon P \to \{-1,1\},
 $$
 and then letting the `subsampled'  dataset $P'$ contain points that were colored $1$, say: 
 $$
 P'=\{p\in P: \chi(p)=1\}. 
 $$ 
We thus would like to find a coloring $\chi$ of $P$ such that
\begin{equation}\label{eq:subsample-disc}
\left|\frac1{|P|}\sum_{p\in P} \sfK(p, q)-\frac{2}{|P|} \sum_{p\in P'} \sfK(p, q)\right|\leq \e \mu.
\end{equation}

 The left hand side of~\eqref{eq:subsample-disc} can be expressed as

\begin{equation*}
\begin{split}
\left|\sum_{p\in P} \sfK(p, q)-2 \sum_{p\in P'} \sfK(p, q)\right|&=\left|\sum_{p\in P} \sfK(p, q)-2 \sum_{p\in P: \chi(p)=1} \sfK(p, q)\right|\\
&=\left|\sum_{p\in P: \chi(p)=-1} \sfK(p, q)-\sum_{p\in P: \chi(p)=1} \sfK(p, q)\right|\\
&=:\text{disc}_\sfK(P, \chi, q),
\end{split}
\end{equation*}
the discrepancy of coloring $\chi$ with respect to query $q$. Choosing $P'$ to be a uniformly random subset of $P$ containing every data point independently with probability $1/2$ amounts to a uniformly random coloring $\chi$ of $P$, and a simple calculation shows that 
$$
\text{disc}_\sfK(P, \chi, q)=O\left(\sqrt{\sum_{p\in P} \sfK(p, q)^2}\right)=O\left(\sqrt{\sum_{p\in P} \sfK(p, q)}\right)=O\left(\sqrt{\mu\cdot |P|}\right)
$$
for every kernel bounded by $1$, with constant probability. Substituting back into~\eqref{eq:subsample-disc} and taking the normalizing factor of $\frac1{|P|}$ into account, we get that the error introduced by subsampling is below $\e \mu$ as long as 
$$
\frac1{|P|}\cdot O\left(\sqrt{\mu\cdot |P|}\right)=O\left(\sqrt{\frac{\mu}{|P|}}\right)\ll \e\mu.
$$
This means that we can keep subsampling\footnote{One can verify that the total error induced by a sequence of recolorings is dominated by the error introduced in the last step -- see Section~\ref{sec:improved-coresets} for more details.} while $|P|\gg 1/(\e^2 \mu)$. This recovers the bound from Chebyshev's inequality\footnote{Of course, this derivation also uses Chebyshev's inequality in bounding the discrepancy of a random coloring; however,  as we show next, it readily generalizes to settings when the coloring is obtained by a more careful method than uniformly random choice.}.

\paragraph{Better colorings via Banaszczyk's theorem~\cite{PT20}.} Instead of selecting the coloring randomly,~\cite{PT20} note that since $\sfK$ is a p.d kernel with $\sfK(p, q)\leq 1$ for all $p, q$, there exists an embedding $\phi$ such that for every $p, q\in \mathbb{R}^d$ one has 
$$
\| \phi(p)\|_2\leq 1, \|\phi(q)\|_2\leq 1
$$
as well as 
$$
\sfK(p, q)=\langle \phi(p), \phi(q)\rangle.
$$
The existence of a coloring $\chi$ such that 
\begin{equation}\label{eq:basic-disc}
\text{disc}_\sfK(P, \chi, q)\leq O(\gamma_2(\sfK))= O(1)\cdot \left(\max_{p\in \mathbb{R}^d} \| \phi(p)\|^2_2\right)=O(1)
\end{equation}
for all $q$ then follows by Banaszczyk's theorem  -- see Theorem~\ref{thm:banaszczyk}.\footnote{Banaszcyk's theorem gives a distribution over (random) colorings which achieves a $O(1)$ discrepancy for each $q$ with very high probability, so in general, there is a $O(\sqrt{\log m})$-factor, where $m$ denotes the number of rows, i.e., queries, of the kernel matrix which one wishes to support.}  Here $\gamma_2(\sfK)$ is the $\gamma_2$-norm of the kernel matrix $\sfK$, which provides a commonly used route for upper bounding discrepancy -- see Section~\ref{sec:improved-coresets} for more details. Substituting this bound into~\eqref{eq:subsample-disc}, we get that the error introduced by subsampling is below $\e \mu$ as long as 
\begin{align}
\frac1{|P|}\cdot O(1)\geq 1/(\e\mu). \label{eq:repeat-halve}
\end{align}
This means that we can keep subsampling while $|P|\gg 1/(\e \mu)$. This is a quadratic improvement on the $\e$-dependence from random sampling, but far short of our goal of $\poly{\log(1/\mu)}/\e$. Furthermore, the bound on discrepancy provided by~\eqref{eq:basic-disc} is {\bf tight} in general.

\paragraph{Our approach: discrepancy bounds for well-separated datasets.}  In order to get around the tightness of the above bound for general datasets, we show that every dataset $P$ can be decomposed into a small number of datasets that are `nice' with respect to any fixed query $q\in \RR^d$. This decomposition is independent of the query, and relies on randomized space partitioning in high dimensions akin to locality sensitive hashing. We then design specialized feature embeddings for each element in this decomposition to establish a significantly stronger upper bound on the discrepancy of the corresponding sub-dataset with respect to $q$.

Our geometric assumptions are inspired by those in Fast Multipole Methods~\cite{fast-multipole}. In the Fast Multipole Methods, the space $\RR^d$ is \emph{deterministically} recursively partitioned with $\ell_\infty$-balls of geometrically decreasing diameter. The benefit of using $\ell_{\infty}$-balls is that they tile $\RR^d$, and whenever $p$ and $q$ belong to two different $\ell_{\infty}$-balls which are not adjacent (share a corner) at a particular radius, $p$ and $q$ are separated by at least the diameter. The downside is that, since Euclidean separation ultimately matters, one must decrease the radius by at least a constant (and, in fact, at least $\sqrt{d}$) factor at every level -- and this leads to an exponential dependence in the dimension as the tree encoding the recursive partition has degree $\exp(\Omega(d))$. Our recursive partitioning will be \emph{randomized} and use $\ell_2$-balls of randomly chosen radii. With such a partitioning scheme one can ensure a weaker separation, namely a relative $\Omega(\frac1{\sqrt{d}})$ separation. We show, however, that this suffices!  Specifically, we show strong discrepancy bounds when either 
\begin{description}
\item[{\bf (1)}] the dataset $P$ lies in a spherical shell and the query $q$ lies in a slightly larger spherical shell (see Fig.~\ref{fig:well-sep})

or

\item[{\bf (2)}] the query $q$ lies inside a spherical shell and dataset $P$ lies in a slightly larger spherical shell (the opposite of {\bf (1)}). 
\end{description}

This in particular is where we crucially use the assumption that the kernel is smooth -- see Remark~\ref{rem:smoothness} in Section~\ref{sec:improved-coresets} for more details.

  	\begin{figure}
	\centering
	\tikzstyle{vertex}=[circle, fill=black!70, minimum size=3,inner sep=1pt]
	\tikzstyle{svertex}=[circle, fill=black!100, minimum size=5,inner sep=1pt]
	\tikzstyle{gvertex}=[circle, fill=green!80, minimum size=7,inner sep=1pt]
	
	\tikzstyle{evertex}=[circle,draw=none, minimum size=25pt,inner sep=1pt]
	\tikzstyle{edge} = [draw,-, color=red!100, very  thick]
	\tikzstyle{bedge} = [draw,-, color=green2!100, very  thick]
	\begin{tikzpicture}[scale=1.7, auto,swap]

	\draw[pattern=north east lines, pattern color=blue] (0,0) circle (2cm);
	\fill[fill=white] (0,0) circle (1.73cm);
	\fill[fill=white] (0,0) circle (1.414cm);
	\draw[pattern=north west lines, pattern color=red] (0,0) circle (1.414cm);
\fill[fill=white] (0,0) circle (1cm);

	\draw (0,0) circle (1cm);
\draw (0,0) circle (1.414cm);
	\draw (1.8,1.8) node {$r_{\max}$};
	\draw[->, >=stealth] (1.7, 1.7) -- (1.5, 1.5);
	
	\draw (2.2,1.1) node {$\rout$};
	\draw[->, >=stealth] (2.1, 1.0) -- (1.6, 0.7);	

	\draw (0.3,0.8) node {$\rmin$};
	\draw (0.5,1.2) node {$\rin$};	
	
	\draw (2.4, -1.8) node {\small dataset $P$};
	\draw[->, >=stealth] (2.4, -1.7) -- (+0.9, -0.9);

	\draw (-2.4, -1.8) node {\small possible queries};
	\draw[->, >=stealth]  (-2.4, -1.7) -- (-1.4, -1.3);

\draw (0,0) circle (1.73cm);
\draw (0,0) circle (2cm);

	
	\node[vertex](v1) at (0, 0) {};  

	\end{tikzpicture}
	\caption{Illustration of the well-separated setup. The point $c \in \RR^d$ is the center of ball, which may be assumed to be the origin after a translation, and we consider two well-separated shells at radii between $[\rmin, \rin]$ and $[\rout, r_{\max}]$ centered at $c$, where $\rin < \rout$. We consider dataset points which lie in the inner shell, with distance between $\rmin$ and $\rin$ from $c$, and queries which lie in the outer shell, with distance between $\rout$ and $r_{\max}$ from $c$. This is the case we consider throughout the technical overview; the symmetric case when the dataset is inside a low radius shell and query is outside will be analogous.} \label{fig:well-sep}
\end{figure}
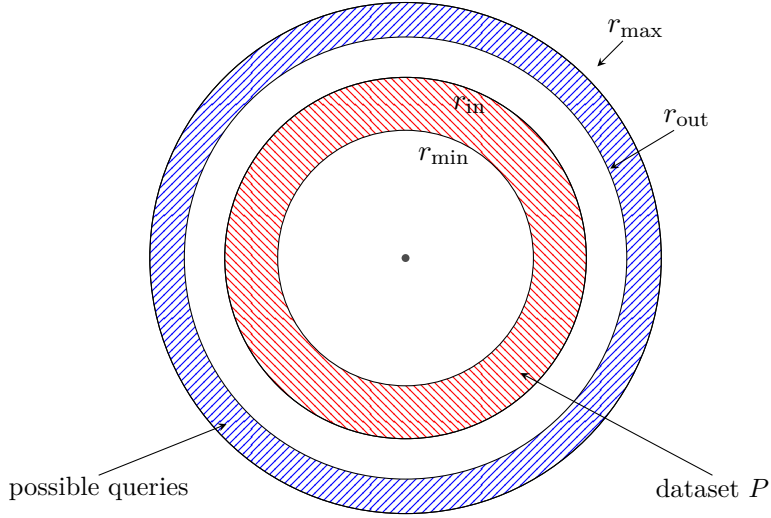

In what follows we give a more detailed outline of how our feature embeddings are constructed, using the $2$-Student kernel as a running example. We first define a basic feature embedding, then explain how to apply discrepancy theory to achieve good colorings via the $\gamma_2$ norm of the embedding and then talk about our modified feature embeddings that achieve strong discrepancy bounds in settings {\bf (1)} and {\bf (2)} above.

\paragraph{Feature Embeddings for Smooth Kernels.} The crucial property of positive definite kernels $\sfK$ is that they may be represented as inner products in a (potentially infinite dimensional) feature space. Consider the following explicit construction for the $2$-Student kernel, which proceeds by taking the inverse Laplace transform of $1 / (1+\|x-y\|_2^2)$ and a Taylor expansion of $e^{x}$:
\begin{align*}
\sfK(x,y) = \dfrac{1}{1 + \|x-y\|_2^2} &= \int_{t:0}^{\infty} e^{- t(\|x-y\|_2^2+1)} dt  = \int_{t:0}^{\infty} e^{-t \|x\|_2^2} e^{-t\|y\|_2^2} e^{2t\langle x, y\rangle} e^{-t} dt \\
						      &= \int_{t:0}^{\infty} e^{-t \|x\|_2^2} e^{-t\|y\|_2^2} \cdot e^{-t} \sum_{k=0}^{\infty} \frac{(2t)^{k}}{k!} \cdot \langle x^{\otimes k}, y^{\otimes k} \rangle dt.
\end{align*}
We may now consider an embedding $\phi$ which takes as input a vector $x \in \Rbb^{d}$ and outputs the (infinite-dimensional) function $\phi(x)$ whose inputs are a number $t \in [0, \infty)$, an index $k \in \Zbb_{\geq 0}$, and sets
\begin{align}
 \phi(x) (t, k) := e^{-t\|x\|_2^2} \sqrt{\dfrac{(2t)^k}{k!}} \cdot e^{-t/2} \cdot x^{\otimes k} \in \Rbb^{d^k}. \label{eq:emb-def}  
 \end{align}
This representation has the benefit that for all $x, y\in \mathbb{R}^d$ 
$$
\sfK(x, y)=\langle \phi(x), \phi(y) \rangle,
$$
where $\langle \phi(x), \phi(y) \rangle=\int_{t:0}^{\infty} \sum_{k=0}^{\infty} (\phi(x)(t, k))^{\intercal} (\phi(y)(t, k)) dt.$

\paragraph{Discrepancy on the Feature Space.} The key to applying discrepancy minimization algorithms is understanding the so-called $\gamma_2$-norm of the kernel matrix, since this will govern the discrepancy that we may achieve (and hence the number of times that we may halve the dataset). Consider the kernel matrix $A = (\sfK(q, p))$ where rows are indexed by a set of possible queries $Q$ and columns are indexed by dataset points $P$. The $\gamma_2$-norm of $A$ is the minimum, over all factorization of $A$ into $U V$, where $U$ is a $|Q| \times d'$ matrix and $V$ is a $d' \times |P|$ matrix, of the maximum row norm of $Q$ times the maximum column norm of $P$. By placing $\phi(q)$ as the rows of $U$ for each potential query in $Q$ and placing $\phi(p)$ as the columns of $V$ for each dataset point in $P$, we obtain that
\[ \gamma_2(A) \leq \max_{q \in Q} \| \phi(q)\|_2 \cdot \max_{p \in P} \| \phi(p) \|_2.\]
By construction, $\|\phi(q)\|_2^2 = \langle \phi(q), \phi(q) \rangle = \sfK(q, q)$ which equals $1$ for any $q$ (and similarly $p$) -- this gives the bound on the $\gamma_2$-norm used by~\cite{PT20}. As we show below, however, significantly stronger upper bounds on the $\gamma_2$-norm can be obtained if the dataset is well-separated from the query in an appropriate way -- this, coupled with a new hashing-based procedure for reducing to the well-separated setting, will lead to our improved bounds.

\paragraph{Modifying the Feature Embedding.} Suppose that every dataset point $p$ in $P$ was contained within a shell of inner radius $\rmin$ and outer radius $\rin$, and that every query $q$ in $Q$ which we will consider was contained within a shell of inner radius $\rout$ (which is larger than $\rin$) and outer radius $r_{\max}$ (the symmetric case when the query is inside a shell close to the origin and the dataset points are in a shell far from the origin will be analogous). See Fig.~\ref{fig:well-sep} for an illustration.

  In Section~\ref{sec:feature-embeddings}, we show that a configuration gives an improved bound on the $\gamma_2$ norm of the kernel matrix. Indeed, the fact that $q$ has norm which is at least $\rout$ and every $p$ in $P$ has norm which is at most $\rin$ guarantees that the points $p$ and $q$ are not too close to each other, i.e., 
\begin{equation}\label{eq:pq-sep}
\|p - q\|_2 \geq \rout - \rin.
\end{equation}
For example, if we could support a small additive error $\xi > 0$ (which we will later incur a logarithmic dependence, so we will set $\xi$ to $\e \mu/n$), then, it suffices to ``cut off'' the feature embedding at 
\[ t_0 := O\left(\frac{\ln(1/\xi)}{(\rout - \rin)^2}\right), \]
 because for any such ``well-separated'' pair of points $p, q \in \Rbb^d$,
\begin{align*}
\int_{t:t_0}^{\infty} e^{-t(\|p-q\|_2^2+1)} dt \leq e^{-t_0(\rout - \rin)^2} \int_{t:t_0}^{\infty} e^{-t} dt \leq \xi.
\end{align*}
Once we introduce this change, we will exploit the fact that $\|p\|_2 \in [\rmin, \rin]$ and $\|q\|_2 \in [\rout, r_{\max}]$ in order to modify the embedding to $\phi'$ such that the norm of $\phi'(p)$ will not increase too much, but the norm of $\phi'(q)$ will decrease a significant amount. Overall, we will show that $\|\phi'(p)\|_2 \|\phi'(q)\|_2$, which gives an upper bound on the $\gamma_2$-norm of the matrix, will be much smaller. In particular, for a setting of $\rho > 1$, we introduce the change 
\begin{equation}\label{eq:disc-ub}
 \phi'(p)(t, k) = \phi(p)(t, k ) \cdot \rho^{k} \quad\text{and}\quad \phi'(q)(t, k) = \frac{\phi(q)(t, k)}{ \rho^k},  
 \end{equation}
and $\phi'$ certifies an improved bound on the $\gamma_2$-norm of an additive $\xi$-perturbation of kernel matrices. In particular, we upper bound the product of $\|\phi'(p)\|_2\|\phi'(q)\|_2$ while using the fact that $\|p\|_2 \in [\rmin, \rin]$ and $\|q\|_2 \in [\rout, r_{\max}]$. We point the reader to Section~\ref{sec:feature-embeddings} with $G(t) = (1+t)^{-1}$, where we show that this product can be at most,
\begin{align} 
\|\phi'(p) \|_2 \| \phi'(q)\|_2 \lesssim \frac{\ln(1/\xi) \cdot \rin^2}{(\rout - \rin)^2 \cdot \rout \cdot \rmin} \leq \frac{\ln(1/\xi) \cdot \rin}{(\rout - \rin)^2 \cdot \rmin}. \label{eq:discrep}
\end{align}
The above bound gives us the upper bound on the discrepancy that we will achieve, and this will dictate how many times we may half the dataset and incur at most $\e \mu$ error. Importantly, since the query $q$ and every dataset point $p$ considered (inside the shell) satisfies by~\eqref{eq:pq-sep}, we have a lower bound on what each point $p$ from the shell contributes to the kernel evaluation to a query $q \in Q$,
$$
\mu \geq \min_{p\in P, q\in Q} \sfK(p, q)\geq \frac1{1+(\rin + r_{\max})^2},
$$
since $\rin + r_{\max}$ is an upper bound on the maximum distance from $q$ to $P$. We will be able to ensure that $r_{\max}=O(\rout)$, $\rin / \rmin = O(1)$, and that $\rout - \rin \geq \Omega(\rout / \sqrt{d})$ (we expand on why in the next subsection). Overall, this means that 
$$
\mu\geq\Omega\left(\frac1{1+\rout^2}\right).
$$
Since we had $\rin / \rmin = O(1)$ and $(\rout - \rin)^2 \geq \rout^2 / d$, when we use the discrepancy bound in (\ref{eq:discrep}), each step of halving incurs error $O(\ln(1/\xi) \cdot \rout^2/d)$, and similarly to the discussion in (\ref{eq:repeat-halve}), we may continue decreasing the dataset until 
\[ \frac{1}{|P|} \cdot O\left( \frac{\ln(1/\xi) \cdot d}{\rout^2} \right) \lesssim \frac{\e}{1+\rout^2}  \qquad \Longrightarrow \qquad |P| \gtrsim  \frac{d\ln(1/\xi)}{\e}.\]


Even though we have specialized the discussion to the 2-Student kernel, we use a characterization of positive definite radial kernels due to Schoenberg in order to use the above embeddings in general. We note that the smoothness assumption come in the following way. Note that the $\gamma_2$-norm bound depends on how $\rmin, \rin$ and $\rout$ relate to each other, which will factor into the number of times we may halve the dataset while incurring at most $\e \mu$ in the error. This must be compared to $\sfK(p, q)$ (which depends on $\| p - q\|_2$) so that the additive error can be absorbed into $\e \cdot \sfK(P, q) / |P|$. (See also, Remark~\ref{rem:smoothness}.)

\begin{rem}
Taking a broader perspective on kernel methods for high-dimensions, we are not aware of any prior work which adapt the feature embeddings to the specific dataset for improved algorithms. As an example, sampling techniques for the Nystr\"{o}m method~\cite{MM17}, random Fourier features~\cite{RR08, AKMMVZ17}, and sketching methods~\cite{ANW14, ACW17, ACW17b, AKKPVWZ20} always consider the feature embedding $\phi \colon \RR^d \to L_2$ which gives $\langle \phi(p),\phi(q)\rangle = \sfK(p, q)$ for all $p, q \in \RR^d$. Our approach, of dividing the dataset and adapting the feature embeddings to the various parts of the dataset, fits nicely within a recent line-of-work on ``data-dependent'' techniques for high-dimensional datasets~\cite{AINR14, AR15, ALRW17, ANNRW18, CKNS20, CJLW22}, and we are hopeful that such techniques are applicable in other algorithmic contexts
\end{rem}

\paragraph{Maintaining Separation for Coresets.} It remains to show that we can build a data structure which always constructs coresets while guaranteeing a separation between queries and dataset points. The idea is to proceed via ball carving of randomly chosen radii. Suppose, for instance, that all dataset points $P$ lie within a ball of radius $R > 0$, and let $c$ denote the center of that ball. We consider three cases, corresponding to where the (unknown) query $q$ may be in comparison to the ball. The first two cases are relatively simple to handle, and most of the work involves the last case.

\paragraph{Case 1 (easy): $q$ is much farther than $R / \e$ from $c$.} Then, since any $p$ lies within distance $R$ from $c$, one may use the triangle inequality to conclude that the kernel evaluation of $\sfK(q, p)$ and $\sfK(q, c)$ is the same up to $1\pm\e$. In this case, a data structure just needs to remember the center $c$ and the size of the dataset $|P|$, so that one can output $|P| \cdot \sfK(q, c)$ to approximate $\sfK(P, q)$. See Definition~\ref{def:far-radius} and Claim~\ref{cl:far-radius}.

\paragraph{Case 2 (relatively easy): $q$ is farther than $3R$ but within $R/\e$ of $c$.} In this case, we can utilize the coreset construction. Indeed, we ``guess'' the distance from $q$ to $c$ (for which there are at most $O(\log(1/\e))$ many choices). Let $R' \geq 3R$ a guessed distance such that $\|q - c\|_2 \approx 3R'$. Pick a randomly chosen point $c'$ drawn from $B_2(c, R'/2)$ uniformly and use that as our new ``center.'' Every point in $P$ will be within a ball of radius $(R'/2 + R) \leq 3R' / 2$ from $c'$ (by the triangle inequality), and at distance at least $R' / 100$ from $c'$ (with high probability for $d = \omega(\log n)$). In addition, the query is within distance at least $5R'/2$ from $c'$. Setting $\rmin = R'/100$, $\rin = 3R'/2$, and $\rout = 5R'/2$, we can always guarantee an upper bound on the $\gamma_2$-norm of $O(1/(R')^2)$. In addition, the $\sfK(q, p) \geq O(1/(R')^2)$, so the repeated halving technique of \cite{PT20} will give the desired coreset. This is handled in Section~\ref{sec:process-and-query-far}.

\paragraph{Case 3 (the difficult case).} This case occurs when $q$ is within $3R$ from $c$ -- see Fig.~\ref{fig:carving} for an illustration. Here, we want to partition the space such that, as in the fast multipole method, we can guarantee some amount of separation, without the exponential dependency obtained from partitioning into cubes (as is usually done in the fast multipole methods). Consider a point $p$ from $P$ and a query $q$ such that $\| p - q \|_2 = \Omega(R)$. We will consider the following randomized partition. First, we sample a random point $c'$ within $O(R)$ from $c$, and then we sample a (random) radius $r$ on the order of $R$. The hope is that $p$ falls inside the ball around $c'$ of radius $r$, and that $q$ falls outside the ball. In order to apply the improved discrepancy bound for well-separated datasets, we must also fit a shell inner radius $\rin$ and outer radius $\rout$, where $\rout - \rin = \Omega(R / \sqrt{d})$. We point the reader to Lemma~\ref{lem:hash}, where we show this partitioning procedure separates each $p \in P$ with $\|p - q\|_2 = \Omega(R)$ with probability at least $\Omega(1/\sqrt{d})$, so after repeating $O(\sqrt{d} \log n)$ times, we are guaranteed to have separated every dataset point $p$ at distance $\Omega(R)$ from $q$. 


	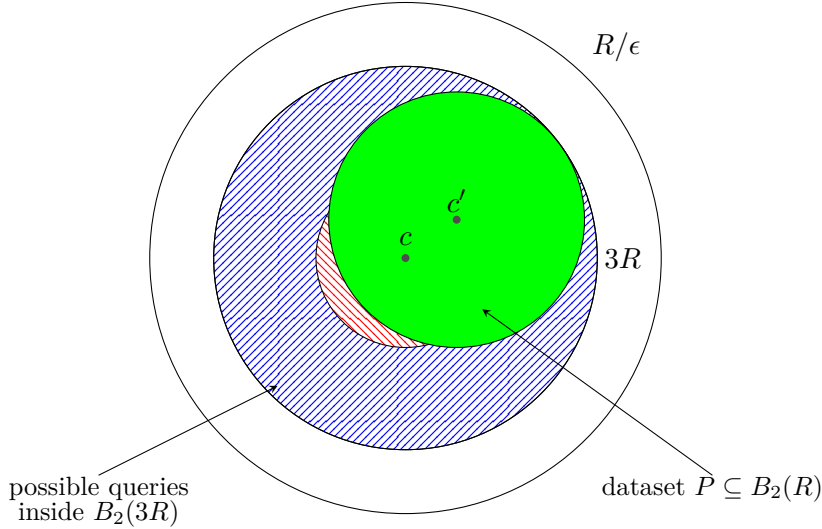
\begin{figure}
	\centering
	\tikzstyle{vertex}=[circle, fill=black!70, minimum size=3,inner sep=1pt]
	\tikzstyle{svertex}=[circle, fill=black!100, minimum size=5,inner sep=1pt]
	\tikzstyle{gvertex}=[circle, fill=green!80, minimum size=7,inner sep=1pt]
	
	\tikzstyle{evertex}=[circle,draw=none, minimum size=25pt,inner sep=1pt]
	\tikzstyle{edge} = [draw,-, color=red!100, very  thick]
	\tikzstyle{bedge} = [draw,-, color=green2!100, very  thick]
	\begin{tikzpicture}[scale=1.7, auto,swap]

	\draw (0,0) circle (2cm);	
	\draw[pattern=north east lines, pattern color=blue] (0,0) circle (1.5cm);
	\fill[fill=white] (0,0) circle (0.7cm);

	\draw[pattern=north west lines, pattern color=red] (0.0,0.0) circle (0.7cm);
	
	\draw[fill=green, opacity=0.5] (0.4,0.3) circle (1cm);
	
	\draw (0.0,0.15) node {$c$};	
	\draw (0.4,0.45) node {$c'$};

	\draw (0,0) circle (1.5cm);
	\draw (1.65,1.65) node {$R/\e$};
	\draw (1.7,0) node {$3R$};
	
	\draw (2.4, -1.8) node {\small dataset $P\subseteq B_2(R)$};
	\draw[->, >=stealth] (2.4, -1.7) -- (+0.6, -0.4);

	\draw (-2.4, -1.8) node {\small possible queries};
	\draw (-2.4, -2.0) node {\small inside $B_2(3R)$};
	
	\draw[->, >=stealth]  (-2.4, -1.7) -- (-1.0, -1.0);


	
	\node[vertex](v1) at (0, 0) {};
	\node[vertex](v2) at (0.4, 0.3) {};	

	\end{tikzpicture}
	\caption{Illustration of ball-carving with balls of radius $O(R)$ when query is at distance at most $3R$ from $c$.} \label{fig:carving}
\end{figure}

In the data structure, we sample a ball and consider the points within the inner-shell and the points outside the outer shell. We construct a coreset for the points inside the inner-shell, and we will use these coresets to evaluate queries which come outside the outer shell (to guarantee separation). Whenever queries come outside the outer shell, the points inside the inner shell are ``captured,'' and the coreset approximates their contribution. For these queries, it remains to recurse on the dataset points which were not captured by the coreset. Similarly, we construct a coreset of the dataset points which evaluate queries which fall inside the inner shell (again, to guarantee separation), and we must recurse on the dataset points outside the outer shell. This recursive partitioning scheme is done in Algorithm~\ref{alg:preprocess}.

Note that there is a small technical issue arising, since the points which fall within radius $\rin$ and $\rout$ from $c'$ are replicated. While this may naively blow up the space of the data structure, the probability that any dataset point falls inside this shell can be made small. The is done by decreasing $\rout - \rin$ to $\alpha R / \sqrt{d}$, at a cost of an increase in the coreset size.

\paragraph{Organization.} In the rest of the paper we first present the formal construction and analysis of our improved coresets for well-separated points in Section~\ref{sec:improved-coresets}.  We then present the details of our reduction to the case of well-separated datasets and the final algorithm in Section~\ref{sec:alg}.


\section{Structural Result: Improved Coresets for Well-Separated Shells}\label{sec:improved-coresets}

\newcommand{\indi}{\mathbf{1}}
\newcommand{\rmax}{r_{\max}}
\newcommand{\rrmax}{\mathsf{r}_{\max}}
\newcommand{\newcen}{\mathsf{newcen}}
\newcommand{\attr}{\mathsf{attr}}
\newcommand{\rrmin}{\mathsf{r}_{\min}}

Given a p.d kernel $\sfK$, the ``kernel trick'' refers to the fact there exists a feature embedding $\phi$ mapping $\Rbb^d$ into a much larger and possibly infinite-dimensional space $\Rbb^{d'}$ where $\sfK(x, y) = \langle \phi(x), \phi(y) \rangle$ for any $x, y \in \Rbb^d$. \ignore{At a high level, one may imagine the ``infinite'' kernel matrix $\sfK$, whose rows and columns are indexed by all points in $\Rbb^d$, and the factorization $\sfK = \Phi \Phi^{\intercal}$, where $\Phi$ has infinitely many rows indexed by all points in $\Rbb^d$ and $d'$ columns such that the $x$-th row of $\Phi$ is $\phi(x)$. }As~\cite{PT20} show, whenever $\sfK(x,x) = 1$ for all $x$, the existence of such a feature embedding $\phi$, as well as discrepancy minimization algorithms are useful for constructing coresets for kernel evaluation. 

The approach proceeds as follows: one receives a dataset $P \subset \Rbb^d$ and wants to support kernel evaluation queries for a finite set of queries $Q \subset \Rbb^d$.\footnote{\cite{PT20} show how to discretize $Q$ to only need to consider $\exp(d)$-sized query sets, but since we will allow queries to fail with a small probability, the discretization will not be an issue in this work.} Then, consider the $|Q| \times |P|$ kernel matrix $A = (a_{qp})$ where the $(q,p)$-entry is $\sfK(q,p)$. Notice that $A \cdot \indi \in \Rbb^{|Q|}$ (where $\indi \in \Rbb^{|P|}$ is the all-$1$'s vector) is the vector of kernel evaluations at each of the queries in $Q$. If one finds a vector $\chi \in \{-1,1\}^{|P|}$ where $\| A \cdot \chi \|_{\infty} \leq \alpha$, then considering the partition of $P$ into $P_{+}$ and $P_{-}$ according to whether a $p \in P$ has $\chi_{p} = 1$ (in which case it belongs to $P_{+}$, and $P_{-}$ otherwise) at least one of the subsets is smaller than $|P|/2$ and for any $q \in Q$, one may approximate the kernel evaluation by only considering that subset. Formally, for both sets $S$ being $P_{+}$ and $P_{-}$,
\[ \left| 2 \sum_{p \in S} \sfK(p, q) - \sum_{p \in P} \sfK(p, q) \right| = \left| \sum_{p \in P_{+}} \sfK(p, q) - \sum_{p \in P_{-}} \sfK(p, q) \right| = (A \chi)_{q} \leq \alpha. \]
In essence, one decreases the size of the dataset by a factor of $2$ and incurs an additive error of $\alpha$ on the kernel evaluation. One can bound $\alpha$ using discrepancy theory.
\begin{defn}[$\gamma_2$-norm of a matrix]
For a matrix $A \in \Rbb^{m \times n}$, the $\gamma_2$-norm of $A$ is given by
\[ \gamma_2(A) = \inf\left\{ \| u_i \|_2 \cdot \| v_j \|_2 : U V = A, \text{ $u_1,\dots, u_m$ are rows of $U$ and $v_1,\dots, v_n$ columns of $V$}\right\}. \]
\end{defn}

\begin{thm}[Banaszczyk's Theorem]\label{thm:banaszczyk}
For any matrix $A \in \Rbb^{m \times n}$, 
\[ \min_{\chi \in \{-1,1\}^n} \| A \cdot \chi \|_{\infty} \lesssim \gamma_2(A) \cdot \sqrt{\log m}. \]
\end{thm}

The final ingredient is showing that $\gamma_2(A) \leq 1$, which is a simple consequence of the existence of $\phi$. In particular, consider the $|Q| \times d'$ matrix $U$ whose rows correspond to $\phi(q)$ for $q \in Q$, and the $d' \times |P|$ matrix whose columns correspond to $\phi(p)$ for $p \in P$. Then, $A = UV$ since $\phi$ is a feature embedding of $\sfK(\cdot, \cdot)$, and the maximum row norm of $U$ and column norm of $V$ is at most $1$, because $\| \phi(q) \|_2^2 = \sfK(q,q) = 1$ and similarly for $p$. 

In this section, we will do two things. 
\begin{enumerate}
\item First, we explicitly construct new feature embeddings for p.d radial kernels $\sfK\colon \Rbb^d \times \Rbb^{d} \to [0, 1]$. For parameters $0 < \rin < \rout$, as well as an additive error $\xi > 0$ (think of $\xi$ as extremely small, since we will depend logarithmically on $1/\xi$), our new feature embedding $\phi$ will only map vectors $z \in \Rbb^d$ which satisfy $\| z\|_{2} \notin [\rin, \rout]$, i.e., they avoid a shell of outer radius $\rout$ and inner radius $\rin$ around the origin. Whenever $\|y \|_2 \leq \rin$ and $\|x\|_2 \geq \rout$, then $\langle \phi(x) ,\phi(y) \rangle$ will be up to $\pm \xi$ the same as $\sfK(x,y)$, yet $\| \phi(x) \|_2 \cdot \|\phi(y)\|_2$ will be smaller than $1$.
\item Second, we use the technique of~\cite{PT20} to repeatedly halve the dataset and construct the coreset whenever a query and dataset will be separated by a shell of inner radius $\rin$ and outer radius $\rout$, and in addition between $\rmin$ and $\rmax$ from the origin. One (minor) difference is that we utilize the self-balancing walk of \cite{ALS21} instead of algorithmic versions of Banaszczyk's theorem. This will lose a logarithmic factor in the size of the coreset, but has the benefit of being very simple.
\end{enumerate}

\subsection{New Feature Embeddings}\label{sec:feature-embeddings}

\newcommand{\Shell}{\mathrm{Shell}}
In what follows, for any two thresholds $0 < \rin < \rout$, we let $\Shell(\rin, \rout) \subset \Rbb^d$ be the set of vectors $x \in \Rbb^d$ with $\rin < \|x\|_2 < \rout$. The main theorem of this section is the following.
\begin{thm}\label{thm:feature-embeddings}
Let $G \colon \Rbb_{\geq 0} \to [0,1]$ be such that $\sfK(x,y) = G(\|x-y\|_2^2)$ is a p.d kernel for every $\Rbb^d$. For any two thresholds $0 < \rin < \rout$ and $\xi > 0$, there exists a map $\phi \colon \Rbb^d \setminus \Shell(\rin, \rout) \to L_2$ such that every $x, y \in \Rbb^d$ where $0 < \| y\|_2 \leq \rin < \rout \leq \|x\|_2$ satisfy:
\begin{itemize}
\item The inner product $\left| \langle \phi(x), \phi(y) \rangle - \sfK(x, y) \right|\leq \xi$.
\item For any $z \in \Rbb^d \setminus \Shell(\rin, \rout)$,
\[ \| \phi(z) \|_2^2 \lesssim G\left( \dfrac{(\rout - \rin)^2}{\ln(1/\xi)} \cdot \frac{\|z\|_2^2}{\rin^2}\right).\]
\end{itemize}
\end{thm}

We will prove Theorem~\ref{thm:feature-embeddings} in the next subsection, but we note below that it directly implies an improvement on the $\gamma_2$-norm of kernel matrices of p.d kernels (up to a small additive error $\xi$). In particular, Theorem~\ref{thm:feature-embeddings} implies that for any sets $P, Q \subset \Rbb^d$ where $P \subset B_2(0, \rin) \setminus B_2(0, \rmin)$ and $Q \not\subset B_2(0, \rout)$ (or vice-versa), there exists a $|Q| \times |P|$ matrix $\tilde{A}$ which is entry-wise $\xi$-close to the $|Q| \times |P|$ kernel matrix $(\sfK(q,p))_{q \in Q, p \in P}$, and
\begin{align*}
\gamma_2(\tilde{A}) \lesssim \left( G\left( \dfrac{(\rout - \rin)^2}{\ln(1/\xi)} \cdot \dfrac{\rout^2}{\rin^2} \right) \cdot G\left( \dfrac{(\rout - \rin)^2}{\ln(1/\xi)} \cdot \frac{\rmin^2}{\rin^2} \right) \right)^{1/2}.
\end{align*}

\begin{rem}[Using Smoothness to Relate $\gamma_2(\tilde{A})$ to $\sfK(P, q) / |P|$]\label{rem:smoothness}
Here, one can see where smoothness of the kernel becomes essential. As per (\ref{eq:repeat-halve}), one can decrease the size of $P$ while the error incurred from discrepancy, which is given by $\gamma_2(\tilde{A}) / |P|$, is smaller than $\e \cdot \sfK(P, q) / |P|$. Note that when $p \in P$ and $q \in Q$ are always at distance at most $2r_{\max}$, we know $\sfK(p, q) \geq G(4r_{\max}^2)$. Thus, we can continue decreasing the dataset while $(1/|P|) \cdot \gamma_2(\tilde{A})$ is significantly smaller than $\e \cdot G(4r_{\max}^2)$, which occurs while
\begin{align*}
|P| \gtrsim \frac{1}{\e} \left( G\left( \dfrac{(\rout - \rin)^2}{\ln(1/\xi)} \cdot \dfrac{\rout^2}{\rin^2} \right) \cdot G\left( \dfrac{(\rout - \rin)^2}{\ln(1/\xi)} \cdot \frac{\rmin^2}{\rin^2} \right) \right)^{1/2} \frac{1}{G(4r_{\max}^2)}. 
\end{align*}
The smoothness comes in when computing the ratio of $G(\cdot)$'s, since the smoothness allows us to relate the $G(r_1^2) / G(r_2^2)$ in terms of $r_2/r_1$.  

In particular, this is why extending our approach to kernels with faster decay requires interesting new ideas: for the Gaussian kernel, for example, the right hand side of the equation above is significantly larger than $\sfK(P, q)/|P|$, as the $\Theta(1/\sqrt{d})$ factor stemming from the separation that our partitioning scheme ensures affects the exponent.
\end{rem}

\subsection{Proof of Theorem~\ref{thm:feature-embeddings}}

We will construct the map $\phi$ explicitly. In order to do so, we first recall Schoenberg's characterization of p.d radial kernels, as well as the Haussdorf-Bernstein-Widder theorem, which we will use in constructing $\phi$. We point to Chapter~7 of \cite{W04} for an extensive treatment of these topics.

\begin{thm}[Schoenberg's Characterization]\label{thm:schoenberg}
A kernel $\sfK(x,y) = G(\|x - y\|_2^2)$ is p.d if and only if the function $G$ is completely monotone on $\Rbb_{\geq 0}$, i.e., $G \in C^{\infty}(\Rbb_{\geq 0})$, and
\[ (-1)^{\ell} \cdot G^{(\ell)}(t) \geq 0 \qquad\text{for all $\ell \in \{0\} \cup \Nbb$ and $t \geq 0$}. \]
\end{thm}

\begin{thm}[Haussdorf-Bernstein-Widder]\label{hbw}
A function $G \colon \Rbb_{\geq 0} \to \Rbb$ is completely monotone if and only if there exists a non-negative finite Borel measure $\mu$ where
\begin{align*}
G(\lambda) = \int_{t:0}^{\infty} e^{-t\lambda} \mu(dt).
\end{align*}
\end{thm}

We introduce the following notation:
\[ t_0 = \dfrac{\log(1/\xi)}{(\rout - \rin)^2} \qquad \text{and}\qquad \rho = 1 - \min\left\{ \dfrac{1}{t_0 \cdot \rin^2}, \frac{1}{2} \right\}. \]
We show how to map each point $\phi(x)$ and $\phi(y)$ when $\|x\|_2 \geq \rout$ and $\|y\|_2 \leq \rin$. First, we let $u_x, v_y \colon [0, t_0] \times (\{0\} \cup \Nbb) \to \Rbb^*$ (where $\Rbb^*$ consists of the union of all finite length tuples $\cup_{j \geq 1} \Rbb^j$) be the functions given by
\begin{align*} 
u_x(t, k) &= e^{-t \|x\|_2^2} \cdot \dfrac{1}{\sqrt{k!}} \left( \sqrt{2t}\cdot \rho\right)^k \cdot \sqrt{\mu(t)} \cdot x^{\otimes k} \in \Rbb^{d^k}, \\
v_y(t,k) &= e^{-t \|y\|_2^2} \cdot \dfrac{1}{\sqrt{k!}} \left( \sqrt{2t} \cdot \frac{1}{\rho}\right)^k \cdot \sqrt{\mu(t)} \cdot y^{\otimes k} \in \Rbb^{d^k}.
\end{align*}
The map $\phi(x)$ will consider the collection of functions $\{ u_{x}(\cdot, k) \colon [0, t_0] \to \Rbb^{d^k} : k \geq 1\}$ and ``concatenate them.'' In particular, we let $H$ denote the Hilbert space over functions $g \colon [0, t_0] \times \left(\cup_{k \geq 1}  [d]^k \right) \to \Rbb$, where $g_1, g_2 \in H$ have
\[ \langle g_1, g_2 \rangle \eqdef \int_{t:0}^{t_0} \sum_{k=0}^{\infty} \sum_{i=1}^{d^k} g_{1}(t, k, i) \cdot g_2(t, k, i) dt. \]
Then, $\phi(x)$ and $\phi(y)$ are the functions where $\phi(x)(t, k, i) = u_{x}(t, k)_i$ and $\phi(y)(t,k, i) = v_y(t, k)_i$. We note that the novelty in the above definitions is introducing the $\rho$ and $1/\rho$ factors in $u_x$ and $v_y$. In the absence of the $\rho$- and $1/\rho$-factors, the proof that we produce would recover features embeddings of unit norm. By introducing these factors, we are able to exploit the fact that $x$ and $y$ have different norms. In particular, these terms cancel out when computing $\langle u_x, v_y\rangle$, but  affect the norms $\| u_x \|_2^2$ and $\|v_y\|_2^2$. For instance, we have
\begin{equation}\label{eq:phi-dot-product}
\begin{split}
\langle \phi(x), \phi(y)\rangle &= \sum_{k=0}^{\infty} \int_{t:0}^{t_0} \langle u_x(t, k), v_y(t, k) \rangle dt \\
	&= \int_{t:0}^{t_0} e^{-t \| x\|_2^2 - t\|y\|_2^2} \sum_{k=0}^{\infty} \frac{1}{k!} \left( 2t \langle x, y \rangle \right)^k \mu(dt) \\
	&= \int_{t:0}^{t_0} e^{-t \| x\|_2^2 - t\|y\|_2^2 + 2t \langle x,y\rangle} \mu(dt) = \int_{t:0}^{t_0} e^{-t \| x - y\|_2^2} \mu(dt).
\end{split}
\end{equation}
By Theorem~\ref{hbw}, we have
\begin{align*}
G(\|x - y\|_2^2) - \langle \phi(x), \phi(y)\rangle  &= \int_{t:t_0}^{\infty} e^{-t\|x-y\|_2^2} \mu(dt) \leq \exp\left(-\dfrac{\ln(1/\xi) \| x-y\|_2^2}{(\rout - \rin)^2} \right) \leq \xi,
\end{align*}
where we used the fact that $\| x - y\|_2 \geq (\rout - \rin)$ by the triangle inequality and the fact that $\int_{t:t_0}^{\infty} \mu(dt)\leq \int_{t:0}^{\infty} \mu(dt)=\sfK(0, 0)\leq 1$.
In particular, $\sfK(x, y) - \xi \leq \langle \phi(x), \phi(y) \rangle \leq \sfK(x, y)$.  

The norms $\| \phi(x) \|_2^2$ and $\| \phi(y) \|_2^2$ satisfy
\begin{equation*}
\| \phi(x) \|_2^2 = \int_{t:0}^{t_0} e^{-2t\|x\|_2^2}  \sum_{k=0}^{\infty} \frac{1}{k!} (2t\rho^2 \|x\|_2^2)^k \mu(dt) \leq \int_{t:0}^{t_0} e^{-2t\|x\|_2^2(1 - \rho^2)} \mu(dt) \leq G\left(2(1-\rho^2)\|x\|_2^2\right) \\
\end{equation*}
and
\begin{equation*}
\begin{split}
\| \phi(y) \|_2^2 &= \int_{t:0}^{t_0} e^{2t(1/\rho^2 - 1)\|y\|_2^2} \mu(dt) = \int_{t:0}^{t_0} e^{3t(1/\rho^2 - 1)\|y\|_2^2} \cdot e^{-t(1/\rho^2-1)\|y\|_2^2}  \mu(dt)   \\
	&\leq e^{3t_0(1 - \rho^2) \|y\|_2^2/\rho^2} \int_{t:0}^{\infty} e^{-t(1 -\rho^2) \|y\|_2^2} \mu(dt) = e^{3t_0(1 - \rho^2) \|y\|_2^2 / \rho^2} \cdot G\left( \dfrac{(1-\rho^2)\|y\|_2^2}{\rho^2}\right).
\end{split}
\end{equation*}
Since $G(\cdot)$ is decreasing (because $G$ is total monotone by assumption, so the derivative of $G$ is always non-positive) and $\rho < 1$, we have
\[ G(2(1-\rho^2) \|x\|_2^2) \leq G\left( \dfrac{\|x\|_2^2}{t_0 \cdot \rin^2} \right) \qquad\text{and}\qquad G\left( \frac{(1-\rho^2) \|y\|_2^2}{\rho^2} \right) \leq G\left(\frac{\|y\|_2^2}{t_0 \cdot \rin^2} \right)\]
The final upper bound on $\|\phi(y)\|_2^2$ comes from the fact
\begin{align*}
\frac{3t_0 (1 -\rho^2) \|y\|_2^2}{\rho^2} \lesssim t_0 \cdot \dfrac{1}{t_0 \cdot \rin^2} \cdot \|y\|_2^2 \leq 1,
\end{align*}
since $\|y\|_2 \leq \rin$. 


\newcommand{\bV}{\mathbf{V}}
\newcommand{\Balance}{\textsc{Balance}}
\newcommand{\bw}{\boldsymbol{w}}
\newcommand{\bsigma}{\boldsymbol{\sigma}}
\newcommand{\calD}{\mathcal{D}}
\newcommand{\bc}{\mathbf{c}}
\newcommand{\boldr}{\mathbf{r}}
\newcommand{\calN}{\mathcal{N}}
\newcommand{\bgamma}{\boldsymbol{\gamma}}
\newcommand{\Prx}{\mathop{\Pr}}
\newcommand{\Ex}{\mathop{\E}}
\newcommand{\bx}{\mathbf{x}}
\newcommand{\eps}{\e}
\newcommand{\boldeta}{\boldsymbol{\eta}}
\renewcommand{\R}{\RR}
\renewcommand{\N}{\mathbb{N}}
\newcommand{\calH}{\mathcal{H}}
\newcommand{\bh}{\mathbf{h}}
\newcommand{\bG}{\mathbf{G}}
\newcommand{\by}{\mathbf{y}}
\newcommand{\bu}{\mathbf{u}}
\newcommand{\BuildTree}{\textsc{BuildTree}}
\newcommand{\calT}{\mathcal{T}}
\newcommand{\calL}{\mathcal{L}}
\newcommand{\calbT}{\boldsymbol{\calT}}
\newcommand{\calbL}{\boldsymbol{\calL}}
\newcommand{\bX}{\mathbf{X}}
\newcommand{\Ber}{\mathrm{Ber}}
\newcommand{\ProcessCaptured}{\textsc{ProcessCaptured}}
\newcommand{\QueryCaptured}{\textsc{QueryCaptured}}
\renewcommand{\bv}{\boldsymbol{v}}
\newcommand{\polylog}{\mathrm{polylog}}
\newcommand{\scT}{\textsc{T}}
\newcommand{\bS}{\mathbf{S}}
\newcommand{\bt}{\mathbf{t}}
\newcommand{\cen}{\mathsf{cen}}
\newcommand{\radius}{\mathsf{rad}}
\newcommand{\OuterBall}{\mathsf{OuterBall}}
\newcommand{\Capt}{\mathrm{Cap}}
\newcommand{\InnerBall}{\mathsf{InnerBall}}
\newcommand{\Child}{\mathsf{Child}}
\newcommand{\rrin}{\mathsf{r}_{\mathrm{in}}}
\newcommand{\rrout}{\mathsf{r}_{\mathrm{out}}}
\newcommand{\cc}{\mathsf{c}}
\newcommand{\sfR}{\mathsf{R}}
\newcommand{\sfr}{\mathsf{r}}
\newcommand{\sfc}{\mathsf{c}}

\subsection{Coresets from New Feature Embeddings}\label{sec:coresets-const}

\begin{defn}[Smooth Function \cite{BCIS18}]
For $L, t \geq 1$, a kernel $\sfK \colon \Rbb^d \times \Rbb^d \to \R$ is $(L, t)$-smooth if for any three points $p_1, p_2, q \in \R^d$ with $p_1 \neq q \neq p_2$, 
\begin{align*}
\max \left\{ \dfrac{\sfK(p_1, q)}{\sfK(p_2, q)},  \dfrac{\sfK(p_2, q)}{\sfK(p_1, q)} \right\} \leq L \cdot \left( \max\left\{\dfrac{\|p_1 - q\|_2}{\|p_2-q\|_2}, \dfrac{\|p_2-q\|_2}{\|p_1 - q\|_2} \right\}\right)^{t}.
\end{align*}
\end{defn}

Suppose the kernel $\sfK(x,y) = G(\|x-y\|_2^2)$ is p.d for every $\R^d$, and in addition, is $(L, t)$-smooth. We will need one preliminary theorem which will follow from the (online) discrepancy minimization algorithm of \cite{ALS21}. Then, we state and prove Lemma~\ref{lem:disc-ds}, which is the main lemma of this section.

\begin{thm}[Self-Balancing Walk~\cite{ALS21}]\label{thm:self-balancing}
For any $n, d \in \N$, there exists a randomized algorithm which receives as input a set of vectors $V = \{ v_1,\dots, v_n\} \in \R^d$ and a parameter $\delta > 0$. The algorithm outputs a (random) subset $\bV' \subset V$ such that, for any vector $u \in \R^d$, with probability at least $1 - \delta$,
\begin{align*}
\left|\sum_{i \in \bV'} \langle v_i, u \rangle - \sum_{i \notin \bV'} \langle v_i, u \rangle \right| \lesssim \log(n/\delta) \cdot \| u\|_2 \cdot \max_{i \in [n]} \| v_i\|_2.
\end{align*}
Furthermore, the algorithm does not require explicit access to $V$; it only requires oracle access to $\{ \langle v_i, v_j \rangle \}_{i,j \in [n]}$.
\end{thm}

\begin{proof}
The statement of Theorem~\ref{thm:self-balancing} above does not explicitly appear in~\cite{ALS21}, but readily follows from the proof of Theorem~1.1 in their paper. In particular,~\cite{ALS21} give a randomized algorithm, $\Balance$, which receives as input a sequence of vectors $v_1,\dots, v_n \in \Rbb^d$ of norm at most $1$, and a failure probability $\delta$.\footnote{We consider the minor modification to their algorithm, where the second condition of Line~4, that $\| w_{i-1}\|_{\infty} \leq c$ is dropped, which allows us to set $c = 30\log(n/\delta)$, as they do in their proof of Theorem~1.2. We note that the reason the check $\| w_{i-1}\|_{\infty} \leq c$ is present in their algorithm is because in \emph{online} discrepancy, one wants to ensure that every coordinate of any partial sum of vectors is small, whereas we will only care about the final vector.} The algorithm produces a sequence of (random) vectors $\bw_0 = 0, \bw_1,\dots, \bw_n \in \Rbb^d$ such that for any vector $u \in \R^d$ with $\|u\|_2\leq1$,\footnote{The specific constant of $240 \pi \log(n/\delta)$ follows their bound of $4Lc$, since $L = 2\pi$ and $c = 30 \log(n/\delta)$.}
\[ \Ex_{\bw_{i}}\left[ \exp\left(\frac{\langle \bw_i, u\rangle^2}{240\pi \log(n/\delta)} \right)\right] \leq \sqrt{2}. \]

In addition, as long as $|\langle \bw_{i}, v_{i+1} \rangle| \leq c = 30 \log(n/\delta)$ for all $i \in [n]$, then there is a setting of signs $\bsigma \in \{-1,1\}^n$ such that every $i \in [n]$ satisfies
\[ \bw_i = \sum_{\ell=1}^i \bsigma_i v_i. \]
As in their proof of Theorem~1.1, one may take a union bound over the $n$ steps and conclude that $|\langle \bw_{n}, u\rangle| \leq c$ except with probability at most $\delta$. Theorem~\ref{thm:self-balancing} stated above simply performs the above argument with $u' = u / \| u \|_2$ and $v_{i}' = v_i / \max_{i} \|v_i\|_2$ to obtain the setting signs $\bsigma \in \{-1,1\}^n$, and sets $\bV'$ to be those indices $i \in [n]$ where $\bsigma_i = 1$.
\end{proof}

The main result of this section is

\begin{lem}\label{lem:disc-ds}
For $L, t \geq 1$, let $\sfK \colon \R^d \times \R^d \to [0, 1]$ be a p.d radial kernel which is $(L, t)$-smooth. There exists a randomized algorithm which receives as input a subset $X \subset \Rbb^d$, four thresholds $0 < \rmin \leq \rin < \rout \leq \rmax$, and three parameters $\eps, \xi, \delta \in (0, 1)$. The algorithm outputs a random subset $\bS \subset X$ and a number $T \geq 0$ which satisfy the following conditions.
\begin{itemize}
\item For any $q \in \Rbb^d$, if $\|q\|_2 \in [\rout , \rmax]$ and $\| x\|_2 \in [\rmin, \rin]$ for every $x \in X$, or $\|q\|_2 \in [\rmin, \rin]$ and $\|x\|_2 \in [\rout, \rmax]$ for every $x \in X$, the following holds with probability at least $1 - \delta$,
\begin{align*}
\left| 2^{T} \sum_{x \in \bS} \sfK(x, q) - \sum_{x \in X} \sfK(x, q) \right| \leq \eps \sum_{x \in X} \sfK(x, q) + 2|X| \xi.
\end{align*}
\item The size of the subset $\bS$ is bounded by
\[ |\bS| \leq \frac{\log^2(|X|/\delta)}{\eps} \cdot L \cdot \left(\dfrac{2\cdot\rmax}{\rout - \rin} \cdot \frac{\rin}{\rmin}\sqrt{\ln(1/\xi)} \right)^t\]
\end{itemize}
\end{lem}

\begin{proof}
Let $\phi \colon \Rbb^d \setminus \Shell(\rmin, \rout) \to L_2$ be the map of Theorem~\ref{thm:feature-embeddings} instantiated with the parameter $\xi$. The algorithm will proceed in iterations and, go through $\ell \in \{0, \dots, T\}$ specifying the sets
\begin{itemize}
\item $\bV_{0} = \{ \phi(x) : x \in X \}$.
\item For $\ell \in [T]$, the set $\bV_{\ell}$ is set to the smallest of $\bV_{\ell-1}'$ or $\bV_{\ell-1} \setminus \bV_{\ell-1}'$, where $\bV_{\ell-1}'$ is the output of algorithm from Theorem~\ref{thm:self-balancing} with input $\bV_{\ell-1}$ and failure probability set to $\delta / T$.\footnote{Even though the vectors $\phi(x)$ are infinite dimensional, there are only a finite set of vectors. In addition, Theorem~\ref{thm:self-balancing} has no dependence on the dimensionality $d$. Thus, it suffices to implicitly work with the subspace spanned by $\{ \phi(x) : x \in X\}$ and $\phi(q)$.}
\end{itemize}
By definition of the procedure above, its simple to see that $|\bV_{T}| \leq |X| / 2^{T}$. In order to show the approximation bound, we first note that
\begin{align}
\left| 2^T \sum_{x \in \bV_T} \sfK(x, q) - \sum_{x \in X} \sfK(x, q) \right| \leq \left|2^{T} \sum_{x \in \bV_T} \langle \phi(x), \phi(q)\rangle - \sum_{x \in X} \langle \phi(x), \phi(q)\rangle \right| + (|X| + 2^T |\bV_T|) \cdot \xi, \label{eq:additive-error}
\end{align}
and because of the fact $|\bV_T| \leq |X| / 2^{T}$, we have $|X| + 2^T|\bV_T| \leq 2|X|$. This handles the claimed additive error. Then, we also have
\begin{align}
\left| 2^T \sum_{x \in \bV_t} \langle \phi(x), \phi(q)\rangle - \sum_{x \in X}  \langle \phi(x), \phi(q)\rangle \right| &\leq \sum_{\ell=0}^{T} 2^{\ell-1} \left| 2 \sum_{x \in \bV_{\ell}} \langle \phi(x), \phi(q)\rangle - \sum_{x \in \bV_{\ell-1}} \langle \phi(x), \phi(q) \rangle  \right| \nonumber \\
		&= \sum_{\ell=0}^{T} 2^{\ell-1} \left|\sum_{x \in \bV'_{\ell-1}} \langle \phi(x), \phi(q)\rangle - \sum_{x \in \bV_{\ell-1} \setminus \bV'_{\ell-1}} \langle \phi(x), \phi(q)\rangle \right| \nonumber \\
		&\lesssim \sum_{\ell=0}^{T} 2^{\ell-1} \cdot \log(|X| T / \delta) \cdot G\left( \dfrac{(\rout - \rin)^2}{\ln(1/\xi)} \cdot \frac{\rmin^2}{\rin^2}\right), \label{eq:disc}
\end{align}
where the final line applies the upper bound of Theorem~\ref{thm:self-balancing}, as well as the fact that 
\[ \| \phi(q) \|_{2} \cdot \max_{x \in X} \| \phi(x)\|_2 \lesssim G\left( \dfrac{(\rout - \rin)^2}{\ln(1/\xi)} \cdot \frac{\rmin^2}{\rin^2}\right). \]
from Theorem~\ref{thm:feature-embeddings}. Furthermore, we notice that $\|x - q\|_2 \leq 2\rmax$ for every $x \in X$, which means $\sfK(x, q) \geq G(4\rmax^2)$. Using the fact that $\sfK$ is $(L, t)$-smooth, we may always upper bound (\ref{eq:disc}) by
\begin{align*}
\left(\dfrac{2^T \log\left( \frac{|X| T}{\delta}\right)}{|X|}\right) \cdot L \cdot \left( \frac{2\cdot\rmax}{\rout - \rin} \cdot \frac{\rin}{\rmin} \cdot \sqrt{\ln(1/\xi)} \right)^t \sum_{x\in X} \sfK(x, q),
\end{align*}
and the above bound is smaller than $\eps$ as long as we set
\[ T = \log_2\left( \frac{\eps \cdot |X|}{\log^2(|X|/\delta)} \cdot \frac{1}{L} \cdot \left( \frac{\rout - \rin}{2\cdot \rmax} \cdot \frac{\rmin}{\rin} \cdot \frac{1}{\sqrt{\ln(1/\xi)}} \right)^t \right),\]
and results in a set $\bS$ of size at most
\[ |\bS| \leq \frac{|X|}{2^{T}} \leq \frac{\log^2(|X|/\delta)}{\eps} \cdot L \cdot \left(\frac{2\cdot \rmax}{\rout - \rin} \cdot \frac{\rin}{\rmin} \cdot \sqrt{\ln(1/\xi)} \right)^t. \]
\end{proof}

The above lemma readily implies the following theorem, which we state so that we can directly invoke this in the subsequent sections. Specifically, Lemma~\ref{lem:disc-ds} can produce a coreset $\bS$, so if we store the coreset points, the parameter $T$, and a translation vector $c \in \R^d$. 

\begin{lem}\label{lem:coreset}

Let $n, d \in \N$, and suppose that for $L, t \geq 1$, $\sfK \colon \Rbb^d \times \Rbb^d \to [0, 1]$ is a p.d radial kernel which is $(L, t)$-smooth. There exist two randomized algorithms with the following guarantees:
\begin{itemize}
\item $\emph{\ProcessCaptured}(X, c, \rmin, \rin, \rout, \rmax, \eps, \xi, \delta)$ receives as input a set $X \subset \R^d$ of size at most $n$, a point $c \in \R^d$, thresholds $\rmin < \rin < \rout \leq \rmax \in \R_{\geq 0}$, error parameters $\eps, \xi \in (0, 1)$, and failure probability $\delta \in (0, 1)$. We are promised that one of the following two hold
\begin{align} 
\|x - c\|_2 &\in [\rmin, \rin] &:\qquad \forall x \in X\label{eq:dataset-inside}\\
\|x - c\|_2 & \in [\rout, \rmax]  &:\qquad \forall x \in X. \label{eq:dataset-outside}
\end{align}
The algorithm outputs a pointer to a data structure $v$.
\item $\emph{\QueryCaptured}(q, v, c, \rmin, \rin, \rout, \rmax, \eps, \xi, \delta)$ receives as input a query $q \in \R^d$, a pointer to a data structure $v$, a point $c \in \R^d$, thresholds $\rmin< \rin< \rout< \rmax \in \R_{\geq 0}$, error parameters $\eps, \xi \in (0, 1)$, and failure probability $\delta \in (0, 1)$. The algorithm outputs a value $\boldeta \in \R_{\geq 0}$. 
\end{itemize}
For any query $q \in \R^d$, if $\|q - c\|_2 \in [\rout, \rmax]$ and (\ref{eq:dataset-inside}) holds, or $\|q - c\|_2 \in [\rmin, \rin]$ and (\ref{eq:dataset-outside}) holds, the following occurs with probability at least $1 - \delta$ over the randomness in the algorithm. We execute $\emph{\ProcessCaptured}(X, c, \rmin, \rin, \rout, \rmax, \eps, \xi, \delta)$ and we let $\bv$ denote the pointer to the data structure it outputs, and we let $\boldeta$ be the output of $\emph{\QueryCaptured}(q, \bv, c, \rmin, \rin, \rout, \rmax, \eps, \xi,\delta)$. Then, 
\begin{itemize}
\item \emph{\textbf{Correctness}}: The estimate $\boldeta \in \Rbb_{\geq 0}$ that we output satisfies
\[ (1-\eps) \sum_{x \in X} \sfK(q,x) - 2\xi|X| \leq \boldeta \leq (1+\eps) \sum_{x \in X} \sfK(q, x) + 2\xi |X|. \]

\item \emph{\textbf{Time and Space Complexity}}: The algorithm $\emph{\ProcessCaptured}(X, c, \rmin, \rin, \rout, \rmax, \eps, \xi, \delta)$ takes time at most $\poly{nd}$ time to output a data structure $\bv$.\footnote{It is important here that the algorithm can efficiently compute $\langle \phi(x), \phi(y) \rangle$ for any $x, y \in X$ for the feature embeddings of Theorem~\ref{thm:feature-embeddings}. See Remark~\ref{rm:dotproduct} for more details.} The total space of $\bv$, as well as the running time of $\emph{\QueryCaptured}(q, \bv, c, \rmin,\rin,\rout, \rmax,\eps, \xi, \delta)$ is, up to a constant factor, at most
\[ d \cdot \frac{\log^2(|X|/\delta)}{\eps} \cdot L \cdot \left( \dfrac{2 \cdot \rmax}{\rout - \rin} \cdot \frac{\rin}{\rmin} \cdot \sqrt{\ln(1/\xi)}\right)^t. \]
\end{itemize}
\end{lem}

\begin{rem}\label{rm:dotproduct}
Lemma~\ref{lem:coreset} assumes black-box access to dot products in the embedded space. This can typically be achieved by obtaining an analytic expression for the measure $\mu$ and integrating as per~\eqref{eq:phi-dot-product}.  For example, if $\sfK(x, y)=1/(1+\|x-y\|^2)$, then $G(\lambda)=1/(1+\lambda)$ in Theorem~\ref{hbw} and one has $\mu(t)=e^{-t}$. Then $\langle \phi(x), \phi(y)\rangle$ can be evaluated per~\eqref{eq:phi-dot-product} at polylogarithmic cost.

\end{rem} 

\newcommand{\Preprocess}{\textsc{Preprocess}}
\newcommand{\Query}{\textsc{Query}}
\newcommand{\ind}{\mathbf{1}}

\section{Algorithmic Result: Data Structure for Evaluating the Coresets}\label{sec:alg}

\subsection{A Ball Carving Hash}\label{sec:hash}

For this section, we present a ball carving hash function. We let $d \in \N$ will be the dimensionality of the space (which we view as a non-constant parameter), as well as a dataset $P \subset \R^d$ of $n$ points. The goal will be to provide a randomized partition of the dataset, such that whenever we use the coreset data structure of \cite{PT20}, we are doing so for points $P$ (and potential queries $Q$) whose kernel matrix has a smaller $\gamma_2$-norm. Hence, this section does not concern the specific kernel $\sfK$, and will simply be a ball-carving hash function.

\begin{lem}\label{lem:hash}
There exists absolute constants $c_1, c_2 > 0$ such that, for any $\alpha  \in (0, c_2)$ and $R > 0$, we have the following. There exists a distribution $\calD$ supported on pairs $(c, r) \in \R^d \times \R_{\geq0}$ which specify a function $h_{c,r} \colon \R^d \to \{0,1,*\}$ given by
\begin{align} 
h_{c,r}(x) =\left\{ \begin{array}{cc} 0 & x \in B_2\left(c, r - \frac{\alpha R}{\sqrt{d}}\right) \\
						     1 & x \notin B_2\left(c, r+\frac{\alpha R}{\sqrt{d}}\right)\\
						     * & \text{o.w.} \end{array} \right. . \label{eq:hash-function}
\end{align}
The distribution satisfies the three guarantees:
\begin{itemize}
\item\label{en:far-points-separated} \emph{\textbf{Separate Far Points:}} For any two points $x, y \in B_2(0, R)$ where $\|x-y\|_2 \geq R/100$, the probability over $(\bc, \boldr) \sim \calD$ that $h_{\bc, \boldr}(x), h_{\bc, \boldr}(y) \in \{0,1\}$ and $h_{\bc, \boldr}(x) \neq h_{\bc, \boldr}(y)$ is at least $c_1/\sqrt{d}$.
\item\label{en:avoid-boundary} \emph{\textbf{Avoid Boundary:}} For any point $x \in B_2(0, R)$, the probability over $(\bc, \boldr) \sim \calD$ that $h_{\bc, \boldr}(x) = *$ is at most $\alpha / \sqrt{d}$.
\item\label{en:far-from-center} \emph{\textbf{Far From Center}:} For any point $x \in \R^d$, with probability at least $1 - 2^{-\Omega(d)}$ over the draw of $(\bc, \boldr) \sim \calD$, $\| x - \bc\|_2 \geq R / 100$, $\|\bc\|_2 \leq 2 R$.
\end{itemize}
\end{lem}

\begin{proof}
The distribution $\calD$ samples a point $\bc$ and a threshold $\boldr$ by letting $\bc \sim \calN(0, R^2 \cdot I_d / d)$ and $\boldr \sim [0, 3R]$. The third item is the simplest: (i) by anti-concentration of $\calN(0, R^2 \cdot I_d / d)$, the probability mass on any ball of radius $R/100$ is at most $2^{-\Omega(d)}$, and (ii) by concentration of $\calN(0, R^2 \cdot I_d / d)$, $\|\bc\|_2 \leq 2R$ except with probability $2^{-\Omega(d)}$. The second item is also simple to argue: for any fixed $c \in \R^d$, $h_{c, \boldr}(x) = *$ whenever $\left| \boldr - \| x - c\|_2 \right| \leq \alpha R / \sqrt{d}$. Therefore, the probability over the draw of $\boldr \sim [0, 3R]$ that the above occurs is at most $\alpha / \sqrt{d}$.  

\newcommand{\bA}{\boldsymbol{A}}
\newcommand{\bB}{\boldsymbol{B}}

We now argue the first item, that $h_{\bc, \boldr}$ tends to separate far points. Consider a fixed setting of $p, q \in B_2(0, R)$, and for a sample $(\bc, \boldr) \sim \calD$, let 
\[ \bgamma \eqdef \big| \|p-\bc\|_2 - \|q - \bc\|_2\big|. \] 
Then, the event that $h_{\bc, \boldr}(p) \neq h_{\bc, \boldr}(q)$ occurs whenever the following two events hold: 
\begin{itemize}
\item Event $\bA$: For $\delta = 1/2$, we have $p, q \in B_2(\bc, (2+\delta)R)$, and 
\item Event $\bB$: The threshold $\boldr$ lies within an interval of $3R$ of length $\bgamma - 2 \alpha R / \sqrt{d}$. 
\end{itemize}
Notice that Event $\bA$ holds with high probability because we may apply the triangle inequality and the fact $p, q \in B_2(0, R)$. If event $\bA$ fails, then the center point $\bc$ must have Euclidean norm larger than $(1+\delta) R$, so
\begin{align*}
\Prx_{\bc}\left[ \text{Event $\bA$ fails}\right] \leq \Prx_{\bc}\left[ \| \bc\|_2 \geq (1+\delta) \cdot R \right] \leq \Prx_{\bx \sim \chi^2(d)}\left[ \bx - d\geq \delta \cdot d\right] \leq e^{-d\delta^2 / 8},
\end{align*}
using a standard concentration inequality on $\chi^2(d)$ random variables (see, Example 2.5 in~\cite{W19}). Therefore, we have
\begin{align}
\Prx\left[ \text{Events $\bA$ and $\bB$ hold}\right] &\geq \Prx\left[ \text{Event $\bB$ holds}\right] - \Prx\left[ \text{Event $\bA$ fails}\right] \nonumber \\
&\geq \Prx\left[ \text{Event $\bB$ holds}\right] - e^{-d\delta^2/8}. \label{eq:hahas}
\end{align}
So it remains to lower bound the probability over $(\bc, \boldr) \sim \calD$ that $h_{\bc, \boldr}(p) \neq h_{\bc, \boldr}(q)$, which is at least
\begin{align}
\Prx\left[ \text{Event $\bB$ holds}\right] &= \Ex_{\bc}\left[ \dfrac{\bgamma - 2 \alpha R / \sqrt{d}}{3R}\right] = \frac{1}{3R} \cdot \Ex_{\bc}\left[ \bgamma \right] - \frac{\alpha}{\sqrt{d}}. \label{eq:split-prob}
\end{align}
We may lower bound the expectation of $\bgamma$ by considering whether or not event $\bA$ holds. In particular, we may always lower bound, for any setting of the randomness of $\bgamma$,
\begin{align}
\bgamma &= \dfrac{\Big| \|p-\bc\|_2^2 - \|q - \bc\|_2^2\Big|}{\|p-\bc\|_2 + \|q-c\|_2} \geq \dfrac{\Big| \|p-\bc\|_2^2 - \|q - \bc\|_2^2\Big|}{2(2+\delta) R} - \ind\left\{ \text{Event $\bA$ fails}\right\} \cdot (R+\|\bc\|_2).   \label{eq:gamma-eq}
\end{align}
The reason is the following: if event $\bA$ holds, then $\|p-\bc\|_2 + \|q -\bc\|_2 \leq 2(2+\delta) R$ and we obtain the desired lower bound. If event $\bA$ fails, then since $\|p\|_2, \|q\|_2 \leq R$,
\[ \dfrac{\Big| \|p-\bc\|_2^2 - \|q - \bc\|_2^2\Big|}{2(2+\delta) R} = \dfrac{\Big| \|p\|_2^2 - \|q\|_2^2 - 2\langle p-q,\bc\rangle\Big|}{2(2+\delta) R} \leq \frac{2R^2}{2(2+\delta) R} + \frac{2\|p-q\|_2\|\bc\|_2}{2(2+\delta) R} \leq R + \|\bc\|_2, \]
so subtracting $\ind\{ \text{Event $\bA$ fails} \} \cdot (R+\|\bc\|_2)$ ensures that the right-hand side of (\ref{eq:gamma-eq}) is negative. Hence,
\begin{align*}
\Ex_{\bc}\left[ \bgamma \right] &\geq \Ex_{\bc}\left[  \dfrac{\big| \| p - \bc\|_2^2 - \| q - \bc\|_2^2\big|}{2(2+\delta) R}\right] - R \cdot \Prx_{\bc}\left[ \text{Event $\bA$ fails}\right] - \Ex_{\bc}\left[ \|\bc\|_2 \cdot \ind\{ \|\bc\|_2 \geq (1+\delta) R\}\right]
\end{align*}
We bound each term:
\begin{align*}
\Ex_{\bc}\left[ \dfrac{\big| \| p - \bc\|_2^2 - \| q - \bc\|_2^2\big|}{2(2+\delta) R} \right] = \frac{1}{2(2+\delta) R} \cdot \Ex_{\bc}\left[ \big| \|p\|_2^2 - \|q\|_2^2 - 2 \langle p - q, \bc \rangle \big| \right] = \Omega \left( \dfrac{\|p-q\|_2}{\sqrt{d}}\right),
\end{align*}
where by 2-stability of the Gaussian distribution, $2 \langle p - q , \bc \rangle$ is distributed like $\calN(0, 4 \|p-q\|_2^2 \cdot R^2 / d)$, and anti-concentration of the Gaussian distribution. Recall that we have already upper bounded the probability that event $\bA$ fails, and finally, we can similarly bound the final expectation,
\begin{align*}
\Ex_{\bc}\left[ \|\bc\|_2 \cdot \ind\{ \|\bc\|_2 \geq (1+\delta) R \} \right] &\leq \int_{\zeta:\delta}^{\infty} R \cdot \Prx_{\bc}\left[ \|\bc\|_2 \geq (1+\zeta) R \right] d\zeta \leq R \int_{\zeta:\delta}^{\infty} e^{-d\zeta^2/8} d\zeta \leq R\cdot e^{-\Omega(d)},
\end{align*}
using the fact $\delta = 1/2$. Substituting into (\ref{eq:split-prob}) into (\ref{eq:hahas}), we obtain
\begin{align*}
\Prx_{(\bc, \boldr) \sim \calD}\left[ \begin{array}{c} h_{\bc,\boldr}(p) \neq h_{\bc,\boldr}(q) \\ h_{\bc, \boldr}(p) , h_{\bc, \boldr}(q) \in \{0,1\} \end{array}\right] &\geq \frac{1}{3R} \left( \Omega\left(\frac{\|p-q\|_2}{\sqrt{d}} \right) - R \cdot e^{-\Omega(d)} \right) - \frac{\alpha}{\sqrt{d}} - e^{-\Omega(d)} \\
&\geq \Omega\left(\frac{1}{\sqrt{d}}\right).
\end{align*} 
since $\|p-q\|_2 \geq R/100$, $\alpha$ is a small enough constant, and $d$ is at least a large enough constant.
\end{proof}

\newcommand{\MEB}{\texttt{MEB}}
\newcommand{\size}{\texttt{size}}
\newcommand{\PreprocessFar}{\textsc{PreprocessFar}}
\newcommand{\FarDS}{\texttt{FarDS}}
\newcommand{\QueryFar}{\textsc{QueryFar}}

\subsection{Data Structure}\label{sec:ds-tree-decomp}

In this section, we will use Lemma~\ref{lem:hash} to preprocess the dataset by recursively hashing. We will consider a dataset $P \subset \R^d$ which consists of $n$ points, an unknown query $q \in \Rbb^d$, and the parameter $\Phi > 1$ which denotes the maximum aspect ratio of $P \cup \{ q\}$, i.e.,
\[ \Phi = \dfrac{\max_{x,y \in P \cup \{ q\}} \| x - y\|_2}{\min_{x\neq y \in P \cup \{q\}} \|x - y\|_2}. \]
We assume that the dimensionality $d$ is $\omega(\log n \log \log \Phi) \leq d$; the assumption is without loss of generality, as we may add coordinates which are set to $0$. 

\begin{thm}\label{thm:main-thm}
For $n,d \in \Nbb$, and $L, t \geq 1$, consider any p.d $(L,t)$-smooth kernel $\sfK \colon \R^d \times \R^d \to [0,1]$. There exists two randomized algorithms $\emph{\Preprocess}$ (Algorithm~\ref{alg:preprocess}) and $\emph{\Query}$ (Algorithm~\ref{alg:query}) such that, 
\begin{itemize}
\item $\emph{\Preprocess}(P, \eps, \xi)$ receives as input a dataset $P \subset \R^d$ of at most $n$ points, and two error parameters $\eps, \xi \in (0, 1)$. The algorithm outputs a pointer to a data structure $u$.
\item $\emph{\Query}(q, u, \eps, \xi)$ receives as input a query $q \in \R^d$, a pointer to a data structure $u$ generated by $\emph{\Preprocess}$, and two error parameters $\eps, \xi \in (0, 1)$. 
\end{itemize}
Suppose that $P \cup \{ q\}$ has aspect ratio at most $\Phi$, then we satisfy the following guarantees with probability at least $0.9$,
\begin{itemize}
\item \emph{\textbf{Correctness}}: If $\bu$ is the output of $\emph{\Preprocess}(P, \eps, \sigma)$ and $\boldeta \in \Rbb_{\geq 0}$ is the output of $\emph{\Query}(q, \bu, \eps, \sigma)$. Then, 
\[ (1-\eps) \sum_{p \in P} \sfK(p, q) - 2\xi n \leq \boldeta \leq (1+\eps) \sum_{p \in P} \sfK(p, q) + 2\xi n.\]
\item \emph{\textbf{Space Complexity}}: The total space of a data structure $\bu$ produced by $\emph{\Preprocess}(P, \eps, \xi)$ is at most, up to a constant factor,
\[ \frac{nd}{\eps} \cdot L \cdot \left( \sqrt{d \cdot \polylog(nd \Phi\ln(1/\xi)/\eps) \cdot \ln(1/\xi)} \right)^t. \]
\item \emph{\textbf{Query Time}}: $\emph{\Query}(q, \bu, \eps, \xi)$ takes time at most, up to a constant factor,
\[ \frac{d}{\eps} \cdot L \cdot \left( \sqrt{d \cdot \polylog(nd\Phi/\eps) \cdot \ln(1/\xi)}\right)^t. \]
\item \emph{\textbf{Preprocessing Time}}: Assuming access to oracles for computing $\langle \phi(x), \phi(y) \rangle$ for $x, y \in \R^d$ and $\phi$ constructed from Theorem~\ref{thm:feature-embeddings}, the algorithm $\emph{\Preprocess}(P, \eps, \xi)$ runs in time $\poly{nd \log \Phi \ln(1/\xi) / \eps}$.
\end{itemize}
\end{thm}

\begin{rem}[Simplifications on the Notation and Error Probability]
In the description of the algorithms $\Preprocess(P, \eps, \xi)$ and $\Query(q, u, \eps, \xi)$, we will make multiple calls to $\ProcessCaptured$ and $\QueryCaptured$ with various parameter settings and we will make the following simplifications. First, we will drop the dependence on $\eps$ and $\xi$ in the description below, as these will remain constant throughout the algorithm. Second, we will, at the forefront, set the error probability $\delta$ in calls to $\ProcessCaptured$ and $\QueryCaptured$ to $1/\poly{dn \Phi/\eps}$. Thus, in the presentation below, we (i) simplify the notation by dropping the dependence on $\delta$ in calls to $\ProcessCaptured$ and $\QueryCaptured$, and (ii) essentially assume that $\ProcessCaptured$ and $\QueryCaptured$ are deterministic algorithms. Indeed, the theorem that we seek to prove, Theorem~\ref{thm:main-thm}, allows for $\polylog(dn\Phi/\eps)$ dependencies, and Lemma~\ref{lem:coreset} has poly-logarithmic dependence on $1/\delta$. Since the preprocessing time of $\Preprocess(P,\eps,\xi)$ and query time $\Query(q, u, \eps, \xi)$ which is at most polynomial in $nd \log \Phi$, setting $\delta \leq 1/\poly{dn\Phi}$ allows one to union bound over all executions of $\ProcessCaptured$ and $\QueryCaptured$ so that they are all correct with high probability.
\end{rem}

The remainder of this section gives the proof of Theorem~\ref{thm:main-thm}. The preprocessing algorithm that we present below can be thought of as consisting of two main parts, where we prepare for an (unknown) query $q \in \Rbb^d$. The data structure will be stored as a binary tree where nodes will hold additional information. Specifically, a call to $\Preprocess(P)$ instantiates a node $u$, and will have certain attributes stored in the node. The notation for these will be $u.\attr$, for some ``attribute'' $\attr$. Each non-leaf node $u$ will have at most two children, which will be stored in $u.\Child(0)$ and $u.\Child(1)$. The formal description appears in Algorithm~\ref{alg:preprocess}.

\paragraph{High Level Structure of Algorithm~\ref{alg:preprocess}} The first thing the algorithm does is enclose the dataset $P$ with an (approximately) minimum enclosing ball around a center point $u.\cen$ of radius $u.\radius$. This is done in Line~\ref{ln:enclose}, and can be done in $O(nd)$ since it will suffice to obtain a constant-factor approximation (for instance, picking an arbitrary point $p \in P$ and letting $u.\cen = p$ and $u.\radius = \max_{p' \in P} \| p - p'\|_2$ suffices). If Line~\ref{ln:check-empty} is triggered, then there is at most $1$ distinct point (which is stored as $u.\cen$) and since the algorithm stores $|P|$ in $u.\size$, it will be able to compute the kernel contribution of $P$. The parameter $u.\sfR$ is set to $2 \cdot u.\radius$ and the algorithm will prepare for two cases: when the (unknown) query $q \in \Rbb^d$ is ``close'' to $u.\cen$ (within distance at most $u.\sfR$), and when the query $q \in \Rbb^d$ is ``far'' from $u.\cen$. 
\begin{itemize}
\item \textbf{Query Close}: In this case, we will be preparing for a query $q$ satisfying $\|q - u.\cen\|_2 \leq u.\sfR$, so that we may instantiate Lemma~\ref{lem:hash} with the origin as $u.\cen$ and $R$ as $u.\sfR$. In this case, we can consider the hash function as dividing $\Rbb^d$ into three parts: (i) within distance $u.\rrin$ of $u.\newcen$ (with high probability, we also have the distance will be at least $u.\rrmin$ as well), (ii) within the shell around $u.\newcen$ of inner-radius $u.\rrin$ and outer-radius $u.\rrout$, and (iii) outside the ball around $u.\newcen$ of radius $u.\rrout$ (and we will also have an upper bound on the distance to $u.\newcen$ of $u.\rrmax$). In this case, we build two coresets with $\ProcessCaptured$ for regions (i) and (iii). Then, we recursively preprocess the points which are not captured by the coresets and store these data structures as the children, $u.\Child(0)$ and $u.\Child(1)$. 
\item \textbf{Query Far}: In this case, we prepare for a query $q$ which satisfies $\|q - u.\cen\|_2 \geq u.\sfR$. Since every point $p \in P$ lies within distance $u.\radius$ of $u.\cen$ and $u.\sfR = 2 \cdot u.\radius$, we can already guarantee a separation of least $u.\radius$ between $q$ and any $p \in P$. Formally, we will handle this in Line~\ref{ln:handle-far} with the sub-routine $\PreprocessFar$ (which we specify shortly). This case will not recursively call $\Preprocess$. 
\end{itemize}
As we will formally show below, recursively applying the ball-carving hash function result in calls to $\Preprocess$ with datasets of decreasing radii. This is because Lemma~\ref{lem:hash} guarantees that we separate far points with at least some probability. Eventually the radius becomes $0$ and Line~\ref{ln:check-empty} ends the recursion. 

 \begin{algorithm}[H]
	\caption{Preprocessing of a dataset $P \subset \R^d$ into a data structure.} \label{alg:preprocess}
	\label{alg:prepro}
	\begin{algorithmic}[1]
		\Procedure{$\textsc{Preprocess}(P)$	
		}{}
		\State Initialize a node $u$. 
		\State\label{ln:enclose} Let $\MEB(P) \subset \Rbb^d$ denote an (approximately) minimum enclosing ball of $P$.
		\State Store the parameters $u.\radius, u.\cen, u.\sfR, u.\size$ as
		\begin{align*} 
		u.\radius &= \text{radius of $\MEB(P) \in \Rbb_{\geq0}$} \\
		u.\cen &= \text{center of $\MEB(P) \in \Rbb^d$} \\
		u.\sfR &= 2 \cdot u.\radius \\
		u.\size &= |P| \in \Nbb.
		\end{align*}
		\If{$u.\radius = 0$}\label{ln:check-empty}
			\State\label{ln:single-point} \Return $u$.
		\EndIf
		\State\label{ln:sample} Sample $(\bc, \boldr) \sim \calD$ from Lemma~\ref{lem:hash} with $R$ set to $u.\sfR$ and $\alpha = c_1 / (4 \log n \log \Phi)^2$.
		\State\label{ln:well-separate} Store $u.\sfc = \bc$ and $u.\sfr = \boldr$ and define the following quantities
		\begin{align*} 
		u.\rrmin &= \frac{u.\sfR}{100} \\
		u.\rrin &= \boldr - \frac{\alpha \cdot u.\sfR}{\sqrt{d}} \\
		u.\rrout &= \boldr + \frac{\alpha \cdot u.\sfR}{\sqrt{d}} \\
		u.\rrmax &=  3 \cdot u.\sfR \\
		u.\newcen &= u.\sfc + u.\cen.
		\end{align*}
		\State We defined the following subsets of $P$ (which do not necessarily partition $P$):
		\[ \Capt(b) = \left\{ p \in P : h_{\bc, \boldr}(p - u.\cen) = 1 - b\right\} ,\qquad\text{for $b \in \{0,1\}$}  \]
		\For{$b \in \{0,1\}$}
		\State\label{ln:inside-ball} Run $\ProcessCaptured(\Capt(b), u.\newcen, u.\rrmin, u.\rrin, u.\rrout, u.\rrmax)$.
		\State Store the output in $u.\InnerBall(b)$.
		\If{$P \setminus \Capt(b) \neq \emptyset$}
		\State\label{ln:child} Execute $\Preprocess(P \setminus \Capt(b))$ and store the output in $u.\Child(b)$.
		
		\EndIf
		\EndFor
		\State\label{ln:handle-far} Run $\PreprocessFar(P,  u.\cen, u.\radius, u.\sfR)$ and store in $u.\FarDS$.
		\State \Return $u$.
		\EndProcedure
	\end{algorithmic}
\end{algorithm}

 \begin{algorithm}[H]
	\caption{Querying a data structure rooted at $u$ with a point $q \in \R^d$. } 
	\label{alg:query}
	\begin{algorithmic}[1]
		\Procedure{$\textsc{Query}(q, u)$	
		}{}
		\If{$\|q - u.\cen\|_2 \leq u.\sfR$}
			\State We evaluate the hash function $h_{u.\sfc, u.\sfr}(q - u.\cen)$, and set $b \in \{0, 1, *\}$ to be its output. 
			\If{$b = *$}
			\State \Return ``fail.''
			\EndIf 
			\State\label{ln:capt} Execute $\QueryCaptured(q, u.\InnerBall(b), u.\newcen, u.\rrmin, u.\rin, u.\rout, u.\rrmax)$ 
		\State\label{ln:recurse} Execute $\Query(q, u.\Child(b))$.
		\State \Return the sum of the outputs of Line~\ref{ln:capt} and Line~\ref{ln:recurse}.
		\EndIf
		\If{$\|q - u.\cen\|_2 > u.\sfR$}
			\State \Return $\QueryFar(q, u.\FarDS, u.\cen, u.\sfR)$.
		\EndIf
		\EndProcedure
	\end{algorithmic}
\end{algorithm}

\subsection{The $\PreprocessFar$ and $\QueryFar$ Algorithms} \label{sec:process-and-query-far}

In this section, we give the descriptions of $\PreprocessFar$ (Algorithm~\ref{alg:preprocess-far}) and  $\QueryFar$ (Algorithm~\ref{alg:query-far}). These are meant to capture cases where the query lies substantially far from the dataset. The approach will be straight-forward: we will partition the (query) space $\Rbb^d$ into geometrically increasing balls (up to a certain point), and maintain a coreset for each of these balls. The remaining piece is to bound how many balls we need. 

\begin{defn}\label{def:far-radius}
Let $\sfK \colon \Rbb^d \times \Rbb^d \to [0, 1]$ be a p.d radial kernel. For any $r \geq 0$ and $\eps, \xi \in (0, 1)$, let $R(r,\eps, \xi) > 0$ denote the minimum over all $R \geq 0$ such that, for all datasets $P \subset \Rbb^d$ where $P \subset B_2(0, r)$, and for all queries $q \in \Rbb^d$ with $\|q\|_2 \geq R$, 
\[ (1-\eps) \sum_{p \in P} \sfK(p, q) - 2\xi |P| \leq |P| \cdot \sfK(0, q) \leq (1+\eps) \sum_{p \in P} \sfK(p, q) + 2\xi |P|.  \]
\end{defn}
\begin{claim}\label{cl:far-radius}
We always have $R(r, \e, \xi)=O(r \ln(1/\xi)/\eps)$.
\end{claim}
\begin{proof}

By Theorem~\ref{hbw} there exists a non-negative finite Borel measure $\mu$ such that
\begin{equation}\label{eq:bhw}
\sfK(p, q) = \int_{t:0}^{\infty} e^{-t \|p-q\|_2^2} \mu(dt)
\end{equation}
for all $p, q\in \RR^d$. We consider the parameter $R = C r \ln(1/\xi) / \eps$, for large enough constant $C$, and we show that enforcing $\| q\|_2 \geq R$ suffices. Note that $R \geq 2r$, and that for any $p\in B_2(0,r)$, the following occurs. If we consider a query $q\in \Rbb^d$ with $\|q \|_2 \geq R$, we always have $\|p-q\|_2 \geq \|q\|_2 - \|p\|_2 \geq \|q\|_2 / 2$. Consider first the case that $t \geq 0$ satisfies, $t \|q\|_2^2\geq 4\ln(1/\xi)$. Then, we will have 
\begin{equation}\label{eq:cutoff}
e^{-t \|p-q\|_2^2}\leq e^{-t \|q\|_2^2/4}\leq \xi.
\end{equation}
On the other hand, if $t \geq 0$ is such that $t \|q \|_2^2 \leq 4 \ln(1/\xi)$, then 
\begin{align*}
t \leq \frac{4 \ln(1/\xi)}{\|q\|_2^2} \leq \frac{4 \eps^2}{C^2 r^2 \ln(1/\xi)} \qquad\text{and}\qquad t \|q\|_2 \leq \frac{4 \eps}{C r},
\end{align*}
which implies the following two inequalities:
\begin{align*}
\dfrac{e^{-t\|p-q\|_2^2}}{e^{-t\|q\|_2^2}} = e^{-t\|p\|_2^2 + 2t \langle p, q\rangle} &\leq e^{2tr \|q\|_2} \leq e^{8\eps / C} \leq 1+\eps/2 \\
e^{-t\|p\|_2^2 + 2t \langle p, q\rangle} &\geq e^{-tr^2 - 2t r \|q\|_2} \geq e^{-4\eps^2/(C^2 \ln(1/\xi)) - 8\eps / C} \geq 1 - \eps/2.
\end{align*}
Putting both cases together, whenever $\|q\|_2 \geq R$, we have
\begin{align*}
\sum_{p \in P} \sfK(p, q) &= \sum_{p \in P} \int_{t:0}^{\infty} e^{-t\|p-q\|_2^2} \mu(dt) \leq \sum_{p \in P} \int_{t:0}^{4\ln(1/\xi)/\|q\|_2^2} e^{-t\|p-q\|_2^2} \mu(dt) + |P| \xi \\
				     &\leq (1+\eps/2) \sum_{p \in P} \int_{t:0}^{4\ln(1/\xi) / \|q\|_2^2} e^{-t \|q\|_2^2} \mu(dt) + |P| \xi \leq (1+\eps/2) |P| \cdot \sfK(0, q) + |P| \xi,
\end{align*}
and analogously, 
\begin{align*}
\sum_{p \in P} \sfK(p, q) &\geq \sum_{p \in P} \int_{t:0}^{4\ln(1/\xi)/\|q\|_2^2} e^{-t\|p-q\|_2^2} \mu(dt) \geq (1-\eps) \sum_{p \in P} \int_{t:0}^{4\ln(1/\xi) / \|q\|_2^2} e^{-t\|q\|_2^2} \mu(dt) \\
					 &\geq (1-\eps/2) \sum_{p \in P} \int_{t:0}^{\infty} e^{-t\|q\|_2^2} \mu(dt) - |P| \xi \geq (1-\eps/2) |P| \cdot \sfK(0, q) - |P| \xi.
\end{align*}
Re-arranging terms gives the desired inequality.
\ignore{
\begin{equation}\label{eq:ratio}
\begin{split}
e^{-t \|p-q\|^2}/e^{-t \|p'-q\|^2}=e^{-t \|p\|^2+t \|p\|^2}\cdot e^{-2t \langle p-p', q\rangle},
\end{split}
\end{equation}

and we next show that the two terms on the rhs are close to $1$. Suppose that $\|q\|\geq R(r, \e, \xi)\geq C(r/\e)\sqrt{\log(1/\xi)}$, as per the assumption of the claim. Note that in that case 
\begin{equation}\label{eq:p-vs-q}
\|p\|\leq r\leq \frac{\e}{C\sqrt{\log 1/\xi}}\cdot \|q\|.
\end{equation}

The exponent of the first factor on the rhs of~\eqref{eq:ratio} satisfies
$$
t \cdot \left|-\|p\|^2+\|p'\|^2\right| \leq t r^2\leq \frac{\e^2}{C^2\log(1/\xi)}\leq \e/3.
$$

The exponent of the second factor on the rhs of~\eqref{eq:ratio} satisfies
\begin{equation*}
\begin{split}
\left|2t \cdot \langle p-p', q\rangle\right|&\leq 2t\cdot \|p-p'\| \|q\|\\
&\leq 4t \|q\|^2\cdot \frac{\e}{C\sqrt{\log (1/\xi)}}\\
&\leq 4 (16\log (1/\xi))\cdot \frac{\e}{C\log (1/\xi)}\\
&\leq 4 (16\log (1/\xi))\cdot \frac{\e}{C\log (1/\xi)}\\
&\leq \e/4\\
\end{split}
\end{equation*}
as long as $C>0$ is larger than an absolute constant.

Substituting the above two bounds into~\eqref{eq:ratio}, we get that 
$$
1-\e\leq e^{-t \|p-q\|^2}/e^{-t \|p'-q\|^2}\leq 1+\e
$$
as long as $\e$ is smaller than an absolute constant. Combining this with~\eqref{eq:cutoff} and then~\eqref{eq:bhw} gives the result.

\xxx[MK]{seems correct to me know, but good for someone to check}}
\end{proof}

 \begin{algorithm}[H]
	\caption{Preprocessing of a dataset $P \subset \R^d$ within a ball when the query will be far.} \label{alg:preprocess-far}
	\begin{algorithmic}[1]
		\Procedure{$\PreprocessFar(P, c, \rin, \sfr)$}{}
		\State\label{ln:find-center} Initialize a node $u$, and sample a point $\bc_0 \sim B_2(0, \rin/3)$ and let
		\begin{align*}
		u.\cen &= c + \bc_0 \in \Rbb^d \\
		u.\rrmin &= \frac{\rin}{100} \\
		u.\rrin &= 4 \cdot \rin / 3 \\
		u.\sfr &= \sfr - \rin / 3 \geq \frac{5}{4} \cdot u.\rrmin \quad \text{(when $\sfr \geq 2\cdot \rmin$)} \\
		u.\size &= |P|.
		\end{align*}
		\For{$h \in \{0, \dots, \lceil \log_2\left(R(u.\sfr,\eps, \xi) / u.\sfr\right) \rceil\}$}
		\State Execute $\ProcessCaptured(P, u.\cen, u.\rrmin, u.\rrin, 2^h \cdot u.\sfr, 2^{h+1} \cdot u.\sfr)$,
		\State Store the output in $u.\OuterBall(h)$. 
		\EndFor
		\State \Return $u$.
		\EndProcedure
	\end{algorithmic}
\end{algorithm}

 \begin{algorithm}[H]
	\caption{Querying of a dataset $P \subset \R^d$ within a ball when the query is far.} \label{alg:query-far}
	\begin{algorithmic}[1]
		\Procedure{$\QueryFar(q, u, c, \sfr)$}{}
		\State Compute $r = \| q - u.\cen\|_2$ (note we will always have $r \geq u.\sfr$).
		\State Let $h \in \{0, 1,2, \dots\}$ such that $2^h \cdot u.\sfr \leq r \leq 2^{h+1} \cdot u.\sfr$. 
		\If{$h \leq \lceil \log_2(R(u.\sfr, \eps, \xi) / u.\sfr)\rceil$}
			\State \Return $\QueryCaptured(q, u.\OuterBall(h), u.\rrmin, u.\rrin, 2^h \cdot u.\sfr , 2^{h+1} \cdot u.\sfr)$.
		\EndIf
		\If{$h > \lceil \log_2(R(u.\sfr, \eps, \xi) / u.\sfr) \rceil$}
			\State \Return $u.\size \cdot \sfK(q, u.\cen)$.
		\EndIf
		\EndProcedure
	\end{algorithmic}
\end{algorithm}

\begin{lem}
Let $n,d \in \Nbb$, and suppose that, for $L, t \geq 1$, $\sfK \colon \Rbb^d \times \Rbb^d \to [0,1]$ is a p.d radial kernel which is $(L,t)$-smooth. There are two randomized algorithms with the following guarantees:
\begin{itemize}
\item $\emph{\PreprocessFar}(P, c, \rin ,\sfr)$ receives as input a set $P \subset \Rbb^d$ of size at most $n$, a point $c$, and two thresholds $\rin$ and $\sfr$, where $\sfr \geq 2 \cdot \rin$. We are promised that every $p \in P$ satisfies $\| p - c\|_2 \leq \rin$. The algorithm outputs a pointer to a data structure $\bu$.
\item $\emph{\QueryFar}(q, u, c, \sfr)$ receives as input a query $q \in \Rbb^d$, and a pointer to a data structure (generated from $\emph{\PreprocessFar}$). We are promised that the query $q$ satisfies $\|q - c\|_2 \geq \sfr$, and the algorithm returns a value $\boldeta \in \Rbb_{\geq 0}$.
\end{itemize}
We satisfy the following guarantees with probability at least $1 - 1/\poly{nd \Phi/\eps}$:
\begin{itemize}
\item \emph{\textbf{Correctness}}: If $\bu$ is the output of $\emph{\PreprocessFar}(P, c, \rin, \sfr)$ and $\boldeta$ is the output of $\emph{\QueryFar}(q, \bu, c, \sfr)$, then
\[ (1-\eps) \sum_{p \in P} \sfK(p, q) - 2\xi |P| \leq \boldeta \leq (1+\eps) \sum_{p \in P} \sfK(p, q) + 2\xi |P|.  \]
\item \emph{\textbf{Preprocessing Time and Space Complexity}}: The algorithm $\emph{\PreprocessFar}$ preprocesses inputs in time at most $\poly{nd\log\Phi \ln(1/\xi)/\eps}$ to output the data structure $\bu$. The total space of $\bu$ is, up to a constant factor, at most
\[ \frac{d}{\eps} \cdot \polylog (dn\Phi\ln(1/\xi)/\eps) \cdot L \cdot \left( O(\sqrt{\ln(1/\xi)}) \right)^t\] 
\item \emph{\textbf{Query Time}}: The algorithm $\emph{\QueryFar}$ outputs an estimate $\boldeta$ in time, up to a constant factor, at most
\[ \frac{d}{\eps} \cdot \log^2(nd\Phi/\eps) \cdot L \cdot \left( O(\sqrt{\ln(1/\xi)} \right)^t\]
\end{itemize}
\end{lem}

\begin{proof}
The purpose of Line~\ref{ln:find-center} is to find an appropriate center which will guarantee some separation between the query and the dataset, and to ensure that no dataset point is too close to the center. In particular, since $\bc_0 \sim B_2(0, \rin/2)$ is drawn randomly and $d = \omega(\log n \log \log\Phi))$, it is not too hard to check that with high probability, the new center $u.\newcen$ and thresholds $u.\rrmin < u.\rrin < u.\sfr$ satisfy that (i) $u.\sfr - u.\rrin \geq u.\rrin / 4$, (ii) every $p \in P$ satisfies $\| p - u.\cen\|_2 \in [u.\rrmin, u.\rrin]$, and (iii) every $q \in \Rbb^d$ with $\|q - c\|_2 \geq \sfr$ will also satisfy $\|q - u.\cen\|_2 \geq u.\sfr$. Thus, the calls to $\ProcessCaptured(P, u.\cen, u.\rrmin, u.\rrin, 2^h \cdot u.\sfr, 2^{h+1} \cdot u.\sfr)$ and\\ $\QueryCaptured(q, u.\OuterBall(h), u.\rrmin, u.\rrin, 2^h \cdot u.\sfr, 2^{h+1} \cdot u.\sfr)$ satisfy the correctness guarantees for queries $q$ with $\|q - u.\cen\|_2 \in [2^h \cdot u.\sfr, 2^{h+1} \cdot u.\sfr]$. If $\|q - u.\cen\|_2 \geq R(u.\sfr, \eps, \xi)$, then by Definition~\ref{def:far-radius}, the entire kernel contribution of $P$ is approximated by $\sfK(q, u.\cen) \cdot u.\size$. The bound on the space and query time follow from Lemma~\ref{lem:coreset}, plugging in $\rmax = 2^{h+1} \cdot u.\sfr$, $\rout = 2^h \cdot u.\sfr$, $\rin = u.\rrin$ and $\rmin = u.\rrmin$
\[ \frac{2 \cdot 2^{h+1} \cdot u.\sfr}{2^{h} \cdot \sfr - u.\rrin} \cdot \frac{u.\rrin}{u.\rrmin} = O(1). \]
 \end{proof}

\subsection{Analysis of $\Preprocess$ and $\Query$}

Before we begin to analyze $\Preprocess(P)$, it is useful to check that the invocations of $\ProcessCaptured$ satisfy the requirements specified in Lemma~\ref{lem:coreset}. We check both cases of $b \in \{0, 1\}$ individually. 
\begin{itemize}
\item \textbf{Case $b = 1$}. By definition of $h_{c, r}$, we have the following guarantees. The set $\Capt(1)$ consists of all points $p \in P$ where $p$ is within $B_2(u.\newcen, u.\rrin)$. Furthermore, by a union bound over $|P|$ points, with probability at least $1 - |P|/2^{\Omega(d)}$, every $p \in \Capt(1)$ also satisfies $\| p - u.\newcen\|_2 \geq u.\sfR / 100 = u.\rrmin$. Thus, the call to $\ProcessCaptured$ with $b = 1$ satisfies Inequality~\ref{eq:dataset-inside} in Lemma~\ref{lem:coreset}, and will handle cases where $\|q - u.\newcen\|_2 \in [u.\rrout, u.\rrmax]$. 
\item \textbf{Case $b = 0$}. In a similar vein, $\Capt(0)$ consists of points whose distance to $u.\newcen$ is at least $u.\rrout$. Furthermore, since $P \subset B_2(u.\cen, u.\radius)$ and by the sampling procedure (recall $u.\sfc \sim \calN(0, u.\sfR^2 I_d / d)$,) we have $\| u.\sfc\|_2 \leq 2 u.\sfR$ with probability $1 - 2^{-\Omega(d)}$, we have $P \subset B_2(u.\newcen, 3u.\sfR)$, so $\|p - u.\newcen\|_2 \leq u.\rrmax$. Thus, the call to $\ProcessCaptured$ with $b=1$ satisfies Inequality~\ref{eq:dataset-outside} and $\| q - u.\newcen\|_2 \in [u.\rrmin, u.\rrin]$ in Lemma~\ref{lem:coreset} with probability at least $1-2^{-\Omega(d)}$. 
Hence, in both $b \in \{0,1\}$ of Line~\ref{ln:inside-ball}, we apply Lemma~\ref{lem:coreset} with
\[\dfrac{2 \cdot \rmax}{\rout - \rin} \cdot \frac{\rin}{\rmin} = O\left( \frac{\sqrt{d}}{\alpha}\right). \]
\end{itemize} 

With those remarks set, we now analyze the space complexity $\Preprocess(P)$. Since we will later bound the space complexity of $\ProcessCaptured$, we will mostly be concerned with the size of the binary tree rooted at $u$ produced by $\Preprocess(P)$. The one challenging aspect in bounding the size of the tree is that the two sets $P \setminus \Capt(0)$ and $P \setminus \Capt(1)$ are not necessarily disjoint. In particular, if $p \in P$ happens to satisfy $h_{\bc, \boldr}(p - u.\cen) = *$, then $p$ is contained in both sets. This will mean that the tree may be super-linear in size (if we are not careful about the setting of $\alpha$). We will first bound the depth of the tree.

\begin{lem}
The total depth of an execution of $\emph{\Preprocess}(P)$ is $O(\sqrt{d} \log(n \log \Phi) \log \Phi)$ with high probability.
\end{lem}

\newcommand{\bT}{\mathbf{T}}

\begin{proof}
We will show that the following occurs with high probability, which in turn implies the desired bound on the depth. Let $u$ be any node of the tree and $P_u \subset P$ such that $u$ was the output of $\Preprocess(P_u)$. Then, over the randomness of the next $s$ levels down the tree, we consider any path $u = \bu_0, \bu_1, \dots, \bu_s$ down the tree where $\bu_{\ell}$ is the child of $\bu_{\ell-1}$, and we have $\bu_{s}.\radius \leq u.\radius / 2$. Suppose that we set $s$ such that the above event occurs with probability at least $1 - 1 / o(\log \Phi)$ (so we can union bound over $O(\log \Phi)$ such events). Then, every $s$ levels of the tree, the values of $\bu.\radius$ decrease by a factor of $2$, and by the definition of the aspect ratio $\Phi$, there are at most $O(s\log \Phi)$ recursive levels before Line~\ref{ln:single-point} always returns.

We now show that we may set $s = \Theta(\sqrt{d}\log(n \log \Phi))$ in the above argument. Consider fixing any two points $p_1, p_2 \in P_u$. Let $\bT$ denote the random sub-tree rooted $u$ which recursively contains children $\bv$ generated  from calls $\Preprocess(P_{\bv})$ in Line~\ref{ln:child} where $p_1, p_2 \in P_{\bv}$. We notice that $\bv.\radius$ is always at most $2 \cdot u.\radius$ (because $P_{\bv} \subset P_u$ and the factor of $2$ occurs because we may choose a different center of $\MEB(P_u)$ and $\MEB(P_{\bv})$), so that in Line~\ref{ln:sample}, the distribution $\calD$ has $\bv.\sfR \leq 4 u.\radius$. Then, the number of nodes at depth $\ell$ of $\bT$ is a (simple) branching process (known as a Galton-Watson process), where in a node which contains $p_1$ and $p_2$ undergoes the following process:

\begin{itemize}
\item By Lemma~\ref{lem:hash}, the fact $R \leq 4 u.\radius$, and $\| p_1 - p_2 \|_2 \geq u.\radius / 2$, we have that with probability at least $c_1 / \sqrt{d}$, the hash $h_{\bc,\boldr}$ satisfies $h_{\bc,\boldr}(p_1 - \bv.\cen), h(p_2 - \bv.\cen) \in \{0 ,1\}$ and $h_{\bc,\boldr}(p_1 - \bv.\cen) \neq h(p_2 - \bv.\cen)$. Thus, the node $\bv$ has \emph{no descendants in} $\bT$ (since we've separated $p_1, p_2$).
\item If $h_{\bc, \boldr}(p_1 - \bv.\cen) = h_{\bc, \boldr}(p_2 - \bv.\cen) = *$, then $p_1$ and $p_2$ are included in both calls to $\Preprocess(P_{\bv} \setminus \Capt(0))$ and $\Preprocess(P_{\bv} \setminus \Capt(1))$ in Line~\ref{ln:child}. In this case, $\bv$ has \emph{two descendants in} $\bT$, but by Lemma~\ref{lem:hash}, this occurs with probability at most $\alpha / \sqrt{d}$. 
\item Finally, if neither of the above two holds, then $\bv$ has one descendant in $\bT$, since $p_1$ and $p_2$ are not separated, but they are also not both replicated.
\end{itemize}
Suppose that we let $d(s)$ denote the probability that a fixed node $u$ has some ancestor at depth $s$ in $\bT$. Then, we may give a recursive upper bound for $d(s)$, since either (i) a node has one child (with probability at most $1 - c_1 / \sqrt{d}$) and that child must have some ancestor at depth $s-1$ in $\bT$, or (ii) a node has two children (with probability at most $\alpha / \sqrt{d}$), and then at least one of them must have an ancestor at depth $s-1$. Thus,
\begin{align*}
d(s) \leq \left(1 - \frac{c_1}{\sqrt{d}} \right) \cdot d(s-1) + \frac{\alpha}{\sqrt{d}} \left(1 - \left(1 - d(s-1)\right)^2 \right) \leq \left(1 - \frac{c_1}{\sqrt{d}} + \frac{2\alpha}{\sqrt{d}}\right) \cdot d(s-1),
\end{align*}
implying
\begin{align*}
\Prx\left[ \bT \text{ has depth $s$}\right] \leq \left(1 - \frac{c_1}{\sqrt{d}} + \frac{2\alpha}{\sqrt{d}}\right)^{s-1}.
\end{align*}
Since $\alpha \leq c_1 / 4$ and we let $s = \Theta(\sqrt{d} \log(n \log \Phi))$, this probability is at most $(n \log \Phi)^{-10}$, so we can union bound over $n^2$ potential pairs $p_1, p_2$ at distance $\| p_1 - p_2\|_2 \geq u.\radius / 2$, and we get the desired bound.
\end{proof}

\begin{lem}
With probability at least $1 - o(1)$, an execution of $\emph{\Preprocess}(P)$ results in a rooted tree of $O(n \sqrt{d} \log(n \log \Phi) \log \Phi)$ nodes.
\end{lem}

\begin{proof}
Suppose that we consider executing $\Preprocess(P)$ while allowing at most $t$ depths of recursion. Then, we bound the expected number of leaf nodes resulting from such a tree, and the total number of nodes will be at most $t$ times that. Since each leaf node contains at least one point, we bound the total number of leaf nodes containing any one point $p \in P$. Any node $u$ generated from call $\Preprocess(P_u)$ where $p \in P_u$ has that either (i) $u$ is a leaf, (ii) the point $p$ satisfies $h_{\bc, \boldr}(p - u.\cen) \in \{0, 1\}$ in Line~\ref{ln:sample} and $p$ lies in a single child of $u$, or (iii) $p$ satisfies $h_{\bc, \boldr}(p - u.\cen) = *$ in Line~\ref{ln:sample}, and then $p$ lies in both children of $u$. Again, this length-$t$ branching process is a simple process, and it is not hard to see that the expected number of leaves containing $p$ within depth $t$ of $u$ is at most $2^{\alpha t / \sqrt{d}}$, and thus the total expected number of leaves at most $n 2^{\alpha t / \sqrt{d}}$. 

By the Markov inequality at a union bound over the depth of $\Preprocess(P, \eps, \sigma)$, it suffices to consider $t = \Theta(\sqrt{d} \log(n\log \Phi) \log \Phi)$, and since $\alpha \leq c_1 / (4 \log n \log \Phi)^2$, $2^{\alpha t / \sqrt{d}} = O(1)$ and $t \cdot O(n)$ gives us the desired bound.
\end{proof}


\begin{lem}
For any fixed query $q \in \R^d$, the following occurs with high probability. Suppose $\bu$ is the output of an execution of $\emph{\Preprocess}(P)$, and $\boldeta \in \R$ is the output of $\emph{\Query}(q, \bu)$. Then, we have:
\begin{align*}
(1-\eps) \sum_{p \in P} \sfK(p, q) - \sigma |P| \leq \boldeta \leq (1+\eps) \sum_{p \in P} \sfK(p, q) + \sigma |P|.
\end{align*}
\end{lem}

\begin{proof}
First, the query $q$ generates a root-to-leaf path given by the executions $\Query(q, \bv)$. In order for the algorithm to avoid outputting ``fail,'' it cannot be the case $h_{u.\sfc, u.\sfr}(q - u.\cen) = *$.  For any fixed node $v$, the probability that this occurs is always at most $\alpha / \sqrt{d}$. Therefore, we note that the probability over the execution of $\Preprocess(P)$ that $\Query(q, \bu)$ outputs ``fail'' sometime within the first $o(\sqrt{d} / \alpha)$ recursive calls is at most $o(1)$. At the same time, the depth of the tree is at most $O(\sqrt{d} \log(n \log \Phi) \log\Phi) = o(\sqrt{d} / \alpha)$ with high probability. Taking a union bound, we have that we reach the end of the tree and we do not output ``fail'' with high probability.

The other necessary event to check is that, if $v$ is a node on the root-to-leaf path in an execution of $\Query(q, u)$, and the call which initialized $v$ was $\Preprocess(P_v)$, then we always have every $p \in P_v$ satisfies $\| p - v.\newcen\|_2 \geq v.\sfR / 100$ and $\| q - v.\newcen\|_2 \geq v.\sfR / 100$. For any fixed node, the probability that this occurs is at least $1 - (n+1) 2^{-\Omega(d)}$, so we can again union bound over the first $O(\sqrt{d} \log(n\Phi) \log \Phi)$ levels of the tree when $d = \omega(\log n \log\log \Phi)$.

Finally, we note that in any execution of $\Query(q, \bu)$, the output is given by a sum of at most $|P|$ executions of $\QueryCaptured$, so that the additive errors $\sigma$ accumulate to at most $\sigma |P|$.
\end{proof}

\bibliographystyle{alpha} 
\bibliography{waingarten}

\end{document}